\documentclass[runningheads]{llncs}

\usepackage{times}
\usepackage{soul}
\usepackage{url}
\usepackage[hidelinks]{hyperref}
\usepackage[utf8]{inputenc}
\usepackage[small]{caption}
\usepackage{graphicx}
\usepackage{amsmath}
\usepackage{booktabs}
\usepackage{algorithm}
\usepackage{algorithmic}
\usepackage{marvosym}

\usepackage{amsthm}
\begin{document}
\title{Controlling Borda Elections by Adding or Deleting either Votes or Candidates: Complete and Top-Truncated Votes}
\titlerunning{Complete and Top-Truncated Votes}
%
\author{Aizhong Zhou\inst{1} \textsuperscript{(\Letter)} \and
Fengbo Wang\inst{2} \and
Jiong Guo\inst{3}}
\authorrunning{A. Zhou et al.}
%
\institute{
Computer Science and Technology, Ocean University of China
\email{{zhouaizhong}@ouc.edu.cn}\\
\and 
Computer Science and Technology, Ocean University of China\\
\email{{wfb}@stu.ouc.edu.cn}\\
\and
Computer Science and Technology, Shandong University\\
\email{JGuo@sdu.edu.cn}}
\maketitle              
%


\begin{abstract}
An election is defined as a pair of a set of candidates $C=\{c_1,\cdots,c_m\}$ and a multiset of votes $V=\{v_1,\cdots,v_n\}$, where each vote is a linear order of the candidates.
The Borda election rule is characterized by a vector $\langle m-1,m-2,\cdots,0\rangle$, which means that the candidate ranked at the $i$-th position of a vote $v$ receives a score $m-i$ from $v$, and the candidate receiving the most score from all votes wins the election.
Here, we consider the control problems of a Borda election, where the chair of the election attempts to influence the election outcome by adding or deleting either votes or candidates with the intention to make a special candidate win (constructive control) or lose (destructive control) the election.
Control problems have been extensively studied for Borda elections from both classical and parameterized complexity viewpoints.
We complete the parameterized complexity picture for Borda control problems by showing W[2]-hardness with the number of additions/deletions as parameter for constructive control by deleting votes, adding candidates, or deleting candidates.
The hardness result for deleting votes settles an open problem posed by Liu and Zhu~\cite{LZ-TCS-2013}.
Following the suggestion by Menon and Larson~\cite{ML-AAMAS-2017}, we also investigate the impact of introducing top-truncated votes, where each voter ranks only $t$ out of the given $m$ candidates, on the classical and parameterized complexity of Borda control problems.
Constructive Borda control problems remain NP-hard even with $t$ being a small constant.
Moreover, we prove that in the top-truncated case, constructive control by adding/deleting votes problems are FPT with the number $\ell$ of additions/deletions and $t$ as parameters, while for every constant $t\geq 2$, constructive control by adding/deleting candidates problems are W[2]-hard with respect to $\ell$.
\end{abstract}

\section{Introduction}
Elections are a commonly used mechanism to achieve preference aggregation and have applications
in multi-agent settings and political domains. Bartholdi et al.~\cite{BTT-MCM-1992} introduced the usage of computational complexity as a barrier to protect elections against different
manipulative actions. See the book chapters~\cite{FR2016,BR-EC-2016}
for recent surveys of related results.
Here, we focus on the control attacks on elections, where an election chair attempts by adding or deleting either votes or candidates to make a special candidate win the election, the {\em constructive control} model~\cite{BTT-MCM-1992}, or lose the election, the {\em destructive control} model~\cite{HHR-AI-2007}.

Complexity results of control problems have been obtained for many election systems such as Plurality,
Condorcet, Approval Voting, Copeland, and Schulze Voting~\cite{FR2016}.
Recently, control problems for the Borda rule have been studied intensively, one of the most prominent natural voting rules.
Hereby, an election consists of a set of candidates $C=\{c_1,\cdots,c_{m-1},p\}$, a multiset $V=\{v_1,\cdots,v_n\}$ of linear orders of candidates, and an integer $\ell\geq 0$.
The Borda rule is characterized by a vector $\langle m-1,m-2,\cdots,0\rangle$.
The candidate ranked at the $i$-th position of a voter receives a score of $m-i$ from $v$, and the candidate receiving the most score from all votes wins the election\footnote{The co-winner model allows multiple candidates with the same highest score co-win the election, while the unique-winner model only allows an unique candidate with the highest score win the election.}.
The Borda control problems asks for applying at most $\ell$ editions to the election such that $p$ wins (constructive control) or loses (destructive control) the election.
The edition operations consider here include vote additions, vote deletions, candidate additions, and candidate deletions.
Edition operations such as partitioning the candidates or votes have also been consider~\cite{ML-AAMAS-2017,NR-AAAI-2017}.
Note that in the vote (candidate) additions control cases, there are two multisets of votes (candidates), one containing registered votes (candidates), and the other unregistered votes (candidates).
The problems ask to turn at most $\ell$ unregistered votes or candidates into registered votes or candidates, such that $p$ wins or loses with respect to all registered votes or candidates.
Hereby, $p$ is a registered candidate and cannot be deleted.
In summary, we study in total $8$ Borda control problems: CCAV, CCDV, CCAC, CCDC, DCAV, DCDV, DCAC, and DCDC, where ``CC'' and ``DC'' stand for ``Constructive Control'' and ``Destructive Control'', respectively.
``AV'', ``DV'', ``AC'', and ``DC'' mean ``Adding Votes'', ``Deleting Votes'', ``Adding Candidates'', and ``Deleting Candidates'', respectively.
While DCAV, DCDV, DCAC, and DCDC of Borda election are proven to be polynomial-time solvable~\cite{FR2016,R-RIT-2007}, their corresponding constructive control versions turn to be NP-hard~\cite{HS-ECAI-2016,NR-AAAI-2017,R-RIT-2007}.
Constructive Borda control problems have also been studied in connection to some special vote structures such as single-peaked or single-dived~\cite{Y-AAMAS-2017}.

In addition to classical complexity, Borda control problems have also been investigated from the viewpoint of parameterized complexity.
Liu and Zhu~\cite{LZ-TCS-2013} proved that CCAV is W[2]-hard with respect to the number of additions and posed the parameterized complexity of CCDV with respect to the number of deletions as an open problem.
Chen et al.~\cite{CFNT-AAAI-2015} established the NP-hardness of CCAC and CCDC with a constant number of votes, ruling out fixed-parameter tractability (FPT) with respect to the number of votes. 
Yang~\cite{Y-ECAI-2014} proved that all control by adding or deleting votes or candidates problems for a class of scoring rules (including Borda) are FPT with the number of candidates as parameter.

Another line of research motivating our study deals with a special type of partial votes: top-truncated votes~\cite{BS-MPCO-2009}.
Since in many practical settings the voters may not be able to determine a complete order over all candidates, it is important to understand the impact of having partial votes on complexity behaviour of election problems.
Most previous studies on partial votes assume that each voter's preference can be extended to a complete order, for instance, the possible and
necessary winner problems~\cite{DJ12,LV11}.
Recent developments in this direction include possible winner on partitioned partial orders~\cite{B-AAMAS-19} or partial chains~\cite{VP21,CDKKRS-20}, controlling Bucklin, Fallback, and approval voting on partial information~\cite{EFRS-JCS-15}, bribery problems of $k$-approval and $k$-veto elections on partial information~\cite{BER-AAMAS-2016}.
The study of voting problems in connection with top-truncated votes (also called top orders) is initiated by Narodytska and Walse~\cite{NW-ECAI-2014}.
A top-truncated vote is a complete linear order of a subset of candidates.
This means that the unranked candidates are all tied and ranked last.
Top-truncated votes are natural in many settings, where an agent is certain about her most preferred candidates, but is indifferent to the remaining ones.
Fitzsimmons and Hemaspaandra~\cite{FH-ICAD-2015} generalized top-truncated votes to weak orders, which allow ties at each position of the votes.
Menon and Larson~\cite{ML-AAMAS-2017} firstly studied Borda election problems under top-truncated votes.
All these works focus on identifying cases where top-truncated votes increase or decrease the complexity of election problems.
Note that top-truncated votes have also been considered for other election systems~\cite{A-A-2010,BFLR-AAMAS-2012}.
The formal definitions of the Borda rule variations for top-truncated votes and the corresponding control problems are given in Section 2.

{\bf Our results:} We first complete the parameterized complexity picture of CCAV, CCDV, CCAC, and CCDC.
We show that CCDV with the number of deleted votes as parameter is
W[2]-hard, settling an open question posed by Liu and Zhu~\cite{LZ-TCS-2013}. By similar reductions, we
can also show W[2]-hardness for CCAC and CCDC with the number of candidate additions
or deletions as parameter.
See Table~\ref{table:complete} for an overview of classical and parameterized complexity results of constructive Borda control problems with complete votes.
Note that the corresponding destructive Borda control problems are solvable in polynomial-time for both complete and top-truncated votes~\cite{FR2016,R-RIT-2007}.

For top-truncated votes, let~$t$ denote the maximal
size of the subsets of candidates, over which the voters specify truncated votes. The dichotomy results from~\cite{ELH-NCAI-2014} directly imply that CCAV and CCDV are NP-hard for every constant $t\geq 3$.
The cases with $t\leq 2$ can be easily solved in polynomial time.
We complement these results with NP-hardness of CCAC and CCDC for every constant $t\geq 2$.
CCAC and CCDC with $t=1$ are trivially polynomial-time solvable.
Concerning parameterized complexity, we prove that in the case of truncated votes, CCAV and CCDV are FPT with the number of additions/deletions and $t$ as parameters, while CCAC and CCDC with top-truncated votes
remain W[2]-hard with respect to the number of candidate additions or deletions. 
The results for the constructive control problems with top-truncated votes are summarized in Table~\ref{table:truncated}.

\begin{table*}
\centering\caption{ Classical and parameterized complexity of constructive Borda control
with complete votes. Here, $n$ denotes the number of votes, $m$ the number of candidates, and~$\ell$ the number of control operations.
Results marked with $\clubsuit$ are due to~\cite{R-RIT-2007},
with $\diamondsuit$ due to~\cite{HS-ECAI-2016},
with $\sharp$ due to~\cite{EFS-AIR-2011},
with $\heartsuit$ due to ~\cite{CFNT-AAAI-2015},
with $\pounds$ due to~\cite{Y-ECAI-2014},
and with $\spadesuit$ due to~\cite{LZ-TCS-2013}.}
\begin{tabular}{ccccc}
\hline
&{}
{Classical complexity}
& \multicolumn{3}{c}{Parameterized complexity}\\
\hline
&{}
&{Parameter $m$}&{Parameter $n$}&{Parameter $\ell$}\\
\hline
{CCAV} & NP-h$(\clubsuit)$ & FPT$(\pounds)$ & FPT$(\heartsuit)$ & W[2]-h$(\spadesuit)$\\
\hline
{CCDV} & {NP-h}$(\diamondsuit)$ & {FPT}$(\pounds)$ & {FPT}$(\heartsuit)$ &{\bfseries{W[2]-h}}\bfseries{[Thm.~\ref{Complete-CCDV-ell}]}\\
\hline
{CCAC} & {NP-h}$(\sharp)$ & {FPT}$(\pounds)$ &{para-NP-h}$(\heartsuit)$ &{\bfseries{W[2]-h}}\bfseries{[Thm.~\ref{Complete-CCAC-ell}]}\\
\hline
{CCDC} & {NP-h}$(\heartsuit)$ & {FPT}$(\pounds)$ & {para-NP-h}$(\heartsuit)$ & {\bfseries{W[2]-h}}\bfseries{[Thm.~\ref{Complete-CCAC-ell}]}\\
\hline
\end{tabular}
\label{table:complete}
\end{table*}

\begin{table*}
\centering\caption{Classical and parameterized complexity of constructive Borda control problems with $t$-truncated votes.
Here, $\ell$ denotes the number of additions or deletions, and $t$ denotes the maximal number of ranked candidates in a vote.
Results marked with $\star$ are due to~\cite{HS-ECAI-2016}. The definitions of ${\tt Borda}_{\uparrow}$, ${\tt Borda}_{\downarrow}$, and ${\tt Borda}_{av}$ are in Section 2. }
\begin{tabular}{cccccccc}
\hline
& \multicolumn{3}{c}{Classical complexity}
& {}
& \multicolumn{3}{c}{Parameterized complexity}\\
\hline
&{${\tt Borda}_{\uparrow}$}&{${\tt Borda}_{\downarrow}$}&\multicolumn{1}{c}{${\tt Borda}_{av}$} & {}
&{${\tt Borda}_{\uparrow}$}&{${\tt Borda}_{\downarrow}$}&{${\tt Borda}_{av}$}\\
\hline
{CCAV}
& \multicolumn{3}{c}{$t\leq 2$: P,\ $t\geq 3:$ NP-h$(\star)$}
& {}
& \multicolumn{3}{c}{{\bfseries FPT w.r.t. $\ell$ and $t$\ [Thm.~\ref{Top-CCAV-ell,t}}]}\\
\hline
{CCDV}
& \multicolumn{3}{c}{$t\leq 2$: P,\ $t\geq 3:$ NP-h$(\star)$}
& {}
& \multicolumn{3}{c}{{\bfseries FPT w.r.t. $\ell$ and $t$\ [Thm.~\ref{Top-CCDV-ell,t}}]}\\

\hline
{CCAC}
& \multicolumn{3}{c}{$t=1$: {\bfseries P\ Thm.~\ref{thm:1-p}},\ $t\geq 2:$ {\bfseries NP-h\ {[Thm.~\ref{thm:top-CCAC-up-w}~-~\ref{top-CCAC-CCDC-av-w}]}}}
& {}
& \multicolumn{3}{c}{$t\geq 2:$ {\bfseries W[2]-h w.r.t. $\ell$\ {[Thm.~\ref{thm:top-CCAC-up-w}~-~\ref{top-CCAC-CCDC-av-w}]}}}\\

\hline
{CCDC}
& \multicolumn{3}{c}{{$t=1$: {\bfseries P\ Thm.~\ref{thm:1-p}},\ $t\geq 2:$ \bfseries NP-h\ {[Thm.~\ref{thm:top-CCAC-up-w}~-~\ref{top-CCAC-CCDC-av-w}]}}}
& {}
& \multicolumn{3}{c}{$t\geq 2:$ {\bfseries W[2]-h w.r.t. $\ell$\ {[Thm.~\ref{thm:top-CCAC-up-w}~-~\ref{top-CCAC-CCDC-av-w}]}}}\\
\hline
\end{tabular}
\label{table:truncated}
\end{table*}

\section{Preliminaries}
We model an election as a pair~$E=(C,V)$, where~$C=\{c_1, \dots , c_m\}$ is the set of candidates
and~$V=\{v_1, \dots , v_n\}$ is the multiset of votes.
The aim of the election is to choose a candidate from $C$ being the winner.
Here, we consider four types of control operations to change the winner when the current winner is not what we expect.
The aim of the control operations is to make a special candidate $p\in  C$ win (constructive control) or lose (destructive control) the election according to the votes in $V$ and the rule.
Let~${\tt{score}}(c, v)$ denote the score that a candidate~$c$ receives from a vote~$v$, and ${\tt{score}}(c, V)=\sum_{v\in V}{\tt score}(c, v)$ denote the total score that~$c$ receives from the votes in $V$.
A candidate $c$ is the unique winner of the election, if $\forall c'\neq c$, ${\tt score}(c, V)> {\tt score}(c', V)$; $c$ is a co-winner of the election, if $\forall c'\in C$, ${\tt score}(c, V)\geq {\tt score}(c', V)$.
Here, we only present the complexity for the unique winner model.
The complexity results for the co-winner model can be derived by similar algorithms or reductions shown here.

\subsection{Complete Votes and Top-truncated Votes}
A vote~$v_i$ is said to be complete, if it is an antisymmetric, transitive, and total order of the candidates in~$C$.
We consider also top-truncated votes. A top-truncated vote is a complete vote of a subset~$C'$ of~$C$, where all candidates in~$C\setminus C'$ are assumed to be tied and ranked after the candidates in~$C'$.
A multiset~$V$ of votes is said to be \emph{$t$-truncated}, if each vote in~$V$ is a top-truncated vote of a subset with at most~$t$ candidates.
Let $|v|$ be the number of ranked candidates in a vote $v$, called the length of $v$.
For example, when $C=\{c_1,c_2,c_3,c_4,c_5\}$ and~$V=\{v_1,v_2,v_3\}$ with~$v_1: c_2 > c_3 > c_4$, $v_2: c_5 > c_1$, and~$v_3: c_1 > c_3 > c_4$.
The length of $v_1$ is 3, $v_2$ is 2 and $v_3$ is 3.
We use~$c(v, i)$ to denote the candidate at the $i$-th position $(1\leq i \leq |v|)$ of the vote~$v$, and ${\tt pos}(v,c)$ to denote the position of the candidate~$c$ in the vote~$v$.
If~$c$ does not occur in~$v$, then~${\tt pos}(v,c)=0$.

\subsection{Borda Rule and Variants}
Given an election with~$m$ candidates, a positional scoring rule is defined by a scoring
vector~$\alpha=\langle\alpha_1, \dots , \alpha_m\rangle$, where~$\alpha_1\geq \dots \geq \alpha_m$. The
\emph{Borda} rule is characterized by the scoring vector~$\alpha=\langle m-1,
m-2, \dots , 0\rangle$.
This means that the candidate $c$ ranked at the $i$-th position of a complete vote~$v$ receives a score of~$m-i$ from~$v$, ${\tt score}(c, v)=m-i$.
And the candidate receiving the most score from all votes wins the election.
Define ${\tt diff}(c_i, c_j, \{v\})={\tt score}(c_i, v)-{\tt score}(c_j, v)$, and ${\tt diff}(c_i, c_j, V)=\sum_{v\in V}{\tt diff}(c_i, c_j, \{v\})$.
An unique Borda-winner $c$ has to satisfy ${\tt diff}(c,c',V)>0$ for all $c'\in C$ and $c\neq c$.
To deal with $t$-truncated votes with~$t\le m$, we consider the following three variants of Borda\footnote{In~\cite{ML-AAMAS-2017}, ${\tt Borda}_{\downarrow}$ and ${\tt Borda}_{av}$ are defined only for the case that all votes have the same length $t$, i.e., $\forall~v\in V: |v|=t$. However, while dealing with votes of length less than $t$, the proofs in~\cite{ML-AAMAS-2017} adopt the definitions we use here.}.
\begin{itemize}
\item \emph{Round up $({\tt Borda}_{\uparrow})$}: The candidate $c$ ranked at the $i$-th position of a vote $v$ with~$i\leq |v|\leq t \leq m$ receives a score of~$m-i$, i.e., ${\tt score}(c, v)=m-i$, while each unranked candidate $c'$ (${\tt pos}(v, c')=0$) receives a score of~$0$, ${\tt score}(c', v)=0$.
\item \emph{Round down $({\tt Borda}_{\downarrow})$}: The candidate $c$ ranked at the $i$-th position of a vote $v$ with~$i\leq |v|\leq t \leq m$ receives a score of~$|v|-i+1$, i.e., ${\tt score}(c, v)=|v|-i+1$, while each unranked candidate $c'$ (${\tt pos}(v, c')=0$) receives a score of~$0$, ${\tt score}(c', v)=0$.
\item \emph{Average score $({\tt Borda}_{av})$}: The candidate $c$ ranked at the $i$-th position of a vote $v$ with~$i\leq |v|\leq t \leq m$ receives a score of~$m-i$, i.e., ${\tt score}(c, v)=m-i$, while each unranked candidate $c'$ (${\tt pos}(v, c')=0$) receives a score of~$\frac{(m-|v|-1)+\cdots+0}{m-|v|}=\frac{m-|v|-1}{2}$, ${\tt score}(c', v)=\frac{m-|v|-1}{2}$.
\end{itemize}

We show an example on the score of each candidate receiving from each vote by the rule of ${\tt Borda}_\uparrow$, ${\tt Borda}_\downarrow$ and ${\tt Borda}_{av}$ as follows.

{\bf Example 1:} Given a set of candidates $C=\{c_1, c_2, c_3, c_4, c_5\}$ and a vote $v_1: c_2 > c_3 > c_4$.
The score of each candidate is:\\
$(1)$ Under the rule of ${\tt Borda}_\uparrow$, ${\tt score}(c_2, v_1)=4$, ${\tt score}(c_3, v_1)=3$, ${\tt score}(c_4, v_1)=2$ and ${\tt score}(c_1, v_1)={\tt score}(c_5, v_1)=0$.\\
$(2)$ Under the rule of ${\tt Borda}_\downarrow$, ${\tt score}(c_2, v_1)=3$, ${\tt score}(c_3, v_1)=2$, ${\tt score}(c_4, v_1)=1$ and ${\tt score}(c_1, v_1)={\tt score}(c_5, v_1)=0$.\\
$(3)$ Under the rule of ${\tt Borda}_{av}$, ${\tt score}(c_2, v_1)=4$, ${\tt score}(c_3, v_1)=3$, ${\tt score}(c_4, v_1)=2$ and ${\tt score}(c_1, v_1)={\tt score}(c_5, v_1)=0.5$.

\vbox{}
For all three variants, the candidate receiving the most score wins the election. 
We say ``the score of $c$ is improved corresponding to $c'$'' if ${\tt diff}(c, c', V)$ is increased after the control operation; ``the score of $c$ is decreased corresponding to $c'$'' if ${\tt diff}(c, c', V)$ is decreased after the control operations; and ``the score of $c$ is unchanged corresponding to $c'$" when the value of ${\tt diff}(c, c', V)$ remains unchanged after the control operations.

\subsection{Problem Definitions}
In the following, we give the definitions of the Borda control problems considered in this paper, following the standard notations from~\cite{FR2016}.
Note that destructive control by adding/deleting votes/candidates problems are polynomial-time solvable for Borda elections~\cite{FR2016,LNRVW-AAMAS-2015}.
The same strategy works also for $t$-truncated votes.
Thus, we consider here only constructive control problems.
Without loss of generality, we assume that~$p$ is not the unique winner of the Borda election before control operations.
We first define the control problems with complete votes.
\begin{quote}
{\bf Constructive Borda Control By Adding Votes} (CCAV-Borda) \\
{\bf Input}: A set~$C$ of candidates, a special candidate~$p\in C$, a multiset~$V$ of registered complete votes over~$C$, a multiset~$V^*$ of unregistered complete votes over~$C$, and an integer~$\ell \ge 0$. \\
{\bf Question}: Can~$p$ be the unique Borda winner of the election by adding at most~$\ell$ votes from~$V^*$ to~$V$?
\end{quote}

\begin{quote}
{\bf Constructive Borda Control By Deleting Votes} (CCDV-Borda) \\
{\bf Input}: A set~$C$ of candidates, a special candidate~$p\in C$,
a multiset~$V$ of complete votes over~$C$, and an integer~$\ell \ge 0$. \\
{\bf Question}: Can~$p$ be the unique Borda winner of the election by deleting at most~$\ell$ votes from~$V$?
\end{quote}

\begin{quote}
{\bf Constructive Borda Control By Adding Candidates} (CCAC-Borda)\\
{\bf Input}: Two sets~$C$ and~$C^*$ of candidates, a special candidate~$p\in C$,
a multiset~$V$ of complete votes over~$C\cup C^*$, and an integer~$\ell \ge 0$.\\
{\bf Question}: Can~$p$ be the unique Borda winner of the election by adding at most~$\ell$ candidates from~$C^*$
to~$C$?
\end{quote}

\begin{quote}
{\bf Constructive Borda Control By Deleting Candidates} (CCDC-Borda)\\
{\bf Input}: A set~$C$ of candidates, a special candidate~$p\in C$, a multiset~$V$
of complete votes over~$C$, and an integer~$\ell \ge 0$.\\
{\bf Question}: Can~$p$ be the unique Borda winner of the election by deleting at most~$\ell$ candidates from~$C$?
\end{quote}

CCDV-Borda has only one set $V$ of votes, i.e., the registered votes, and asks for deleting at most $\ell$ votes from $V$ to make $p$ the unique winner.
CCDC-Borda has only one set $V$ of votes and one set $C$ of candidates and asks for deleting at most $\ell$ candidates from $C$ to make $p$ the unique winner.
In the case of CCAC-Borda, there are one set $V$ of votes and two sets of candidates, i.e., the set $C$ of registered candidates and the set $C^*$ of unregistered candidates, where $p$ is always a registered candidate.
The vote multiset $V$ consists of complete linear orders of all candidates in $C\cup C^*$.
Before adding candidates from $C^*$ to $C$, the score vector is given by $\langle |C|-1,|C|-2,\cdots,0\rangle$, and we consider the votes resulted by removing all candidates in $C^*$ from the votes in $v$.
After adding a set $C^{**}$ of candidates in $C^*$ to $C$, the votes become the linear orders resulting by removing the candidates in $C^*\setminus{C^{**}}$ from the input votes in $v\in V$, and the scoring vector becomes $\langle |C\cup C^{**}|-1, |C\cup C^{**}|-2 ,\cdots,0 \rangle$.
In CCDC-Borda, we change the scoring vector accordingly, that is, deleting candidates makes the scoring vector shorter.

The corresponding control problems with $t$-truncated votes, denoted as $t$-CCAV-$\delta$, $t$-CCDV-$\delta$, $t$-CCAC-$\delta$, and $t$-CCDC-$\delta$, replace in the above definition ``complete'' by ``$t$-truncated'' and ``Borda'' by the corresponding variant $\delta$ with~$\delta\in\{$Borda$_\uparrow$, Borda$_\downarrow$, Borda$_{\text av}\}$.
\begin{quote}
{\bf Constructive Control with $t$-Truncated Votes By Adding Votes under the rule of $\delta$} ($t$-CCAV-$\delta$) \\
{\bf Input}: A set~$C$ of candidates, a special candidate~$p\in C$,
a multiset~$V$ of registered $t$-truncated votes over~$C$,
a multiset~$V^*$ of unregistered $t$-truncated votes over~$C$,
and an integer~$\ell \ge 0$. \\
{\bf Question}: Can~$p$ be the unique winner of the election by adding at most~$\ell$ votes from~$V^*$ to~$V$?
\end{quote}

\begin{quote}
{\bf Constructive Control with $t$-Truncated Votes By deleting Votes under the rule of $\delta$} ($t$-CCDV-$\delta$) \\
{\bf Input}: A set~$C$ of candidates, a special candidate~$p\in C$,
a multiset~$V$ of $t$-truncated votes over~$C$, and an integer~$\ell \ge 0$. \\
{\bf Question}: Can~$p$ be the unique winner of the election by deleting at most~$\ell$ votes from~$V$?
\end{quote}

\begin{quote}
{\bf Constructive Control with $t$-Truncated Votes By Adding Candidates under the rule of $\delta$} ($t$-CCAC-$\delta$)\\
{\bf Input}: Two sets~$C$ and~$C^*$ of candidates, a special candidate~$p\in C$,
a multiset~$V$ of $t$-truncated votes over~$C\cup C^*$, and an integer~$\ell \ge 0$.\\
{\bf Question}: Can~$p$ be the unique winner of the election by adding at most~$\ell$ candidates from~$C^*$
to~$C$?
\end{quote}

\begin{quote}
{\bf Constructive Control with $t$-Truncated Votes By Deleting candidates under the rule of $\delta$} ($t$-CCDC-$\delta$)\\
{\bf Input}: A set~$C$ of candidates, a special candidate~$p\in C$, a multiset~$V$
of $t$-truncated votes over~$C$, and an integer~$\ell \ge 0$.\\
{\bf Question}: Can~$p$ be the unique winner of the election by deleting at most~$\ell$ candidates from~$C$?
\end{quote}

In $t$-CCAC-$\delta$, each vote ranks at most $t$ candidates from both registered and unregistered candidate sets.
Therefore, in each vote, ranked candidates cannot become unranked by adding candidates and unranked candidates cannot become ranked by deleting candidates, clearly distinguish $t$-CCAC-$\delta$ and $t$-CCDC-$\delta$ from the control by adding/deleting candidates problems for plurality and $r$-approval voting.
Example 2 and 3 show an example of CCDV-Borda and $t$-CCAC-${\tt Borda}_{\uparrow}$, respectively.

\vbox{}
{\bf Example 2:} Given an instance of CCDV-Borda $(E=(C, V), \ell)$ where $C=\{c_1, c_2, c_3, p\}$, $V=\{v_1, v_2, v_3\}$, $\ell=1$, and $v_1: c_1>p>c_2>c_3$, $v_2:p>c_3>c_1>c_2$, $v_3:c_1>c_2>c_3>p$.
The score of each candidate is: ${\tt score}(c_1, V)=3+1+3=7$, ${\tt score}(c_2, V)=1+0+2=3$, ${\tt score}(c_3, V)=0+2+1=3$, and ${\tt score}(p, V)=2+3+0=5$.
$c_1$ is the winner of the election.
To make $p$ be the unique winner of the election, we can delete the vote $v_3$.
After the deletion operation, the score of each candidate turns to be:
${\tt score}(c_1, V\setminus{v_3})=4$, ${\tt score}(c_2, V\setminus{v_3})=1$, ${\tt score}(c_3, V\setminus{v_3})=2$, and ${\tt score}(p, V\setminus{v_3})=5$.
$p$ becomes the unique winner of the election.
Therefore, $\{v_3\}$ is a solution of CCDV-Borda $(E=(C, V), \ell)$.

\vbox{}
{\bf Example 3:} Given an instance of $3$-CCAC-${\tt Borda}_\uparrow$ $(E=(C, V), \ell)$ where $C=\{c_1, c_2, p\}$, $C^*=\{v_1, v_2, v_3\}$, $\ell=2$ and $v_1: c_1>p>c_5$, $v_2: p>c_3>c_2$, $v_3: c_5>c_4>c_1$.
Before adding candidates from $C^*$ to $C$, $v_1, v_2$ and $v_3$ are $c_1>p, p>c_2$ and $c_1$, respectively.
Now, the score of each candidate is: ${\tt score}(c_1, V)=2+0+2=4$, ${\tt score}(c_2, V)=0+1+0=1$, and ${\tt score}(p, V)=1+2+0=3$.
$c_1$ is the winner of the election.
To make $p$ be the unique winner of the election, we can add the candidates $c_4$ and $c_5$.
After addition operation, $v_1, v_2$ and $v_3$ turn to be $c_1>p>c_5$, $p>c_2$, and $c_5>c_4>c_1$, respetively.
And the score of each candidate becomes: ${\tt score}(c_1, V)=4+0+2=6$, ${\tt score}(c_2, V)=0+3+0=3$, ${\tt score}(c_4, V)=0+0+3=3$, ${\tt score}(c_5, V)=2+0+4=6$, and ${\tt score}(p, V)=3+4+0=7$.
$p$ is the unique winner of the election.
Therefore, $\{c_4, c_5\}$ is a solution of $3$-CCAC-${\tt Borda}_\uparrow$ $(E=(C, V), \ell)$.

\subsection{Parameterized Complexity}
The main hierarchy of parameterized complexity classes is
FPT $\subseteq$ W[1]$\ \subseteq$ W[2] \dots $\subseteq$ W[$h$] \dots $\subseteq$ XP.
A problem is FPT (stands for ``fixed-parameter tractable''), if it
admits an~$O(f(k) \cdot |\mathcal{I}|^{O(1)})$-time
algorithm\footnote{The time is often written as $O^*(f(k))$.}, where~$\mathcal{I}$ is the input,~$k$ is the parameter, and~$f$ can be any computable function.
The classes~W[$h$] with~$h\ge 1$ refer to fixed-parameter intractability. For more details on parameterized complexity we refer to~\cite{CFKLMPPS-Springer-2015}.

The NP-hardness and the fixed-parameter intractable results are achieved by reductions from the W[2]-hard {\bf {\sc Dominating Set}} \cite{DF99}.
Here, the input consists of an undirected graph~$\mathcal{G}=(\mathcal{V},\mathcal{E})$ and an integer~$k\ge 0$.
We seek for a vertex subset~$DS$, such that~$|DS|\le k$, and~$\forall v\in \mathcal{V}$,~$v\in DS$ or~$ \exists u\in DS, \{v,u\}\in \mathcal{E}$.

\section{Parameterized Complexity for Complete Votes}
We first consider complete votes and complete the parameterized complexity picture of Borda control problems by showing W[2]-hardness of CCDV-Borda, CCAC-Borda, and CCDC-Borda with the number of edition operations as parameter.
Hereby, the W[2]-hardness of CCDV-Borda settles an open question in ~\cite{LZ-TCS-2013}.
\begin{theorem}
\label{Complete-CCDV-ell}
CCDV-Borda is W[2]-hard with respect to the number of deletions~$\ell$.
\end{theorem}
\begin{proof}
We prove this theorem by a reduction from {\sc Dominating Set}.
{\sc Dominating Set} is W[2]-hard with respect to the size of the dominating set.
Given a {\sc Dominating set} instance~$(\mathcal{G}=(\mathcal{V},\mathcal{E}), k)$ with~$n'=|\mathcal{V}|$, we use~$N[v_i]$ to denote the closed neighborhood of~$v_i\in \mathcal{V}$ and set~$d_i=|N[v_i]|$.
Without loss of generality, we assume $n'>k\geq 1$ and $d_i\geq 2$ for all $v_i\in V$.
In the following, we construct an equivalent CCDV-Borda instance~$(E=(C,V),\ell)$. 
To this end, we first construct an auxillary election~$(C',A)$ with $(C',A)$ with $C'=C_1 \cup X \cup Y \cup \{c_h,c_w,p\}$.
For each vertex $v_i\in \mathcal{V}$, we create the following candidates:\begin{itemize}
    \item one candidate $c^1_i$ in $C_1$;
    \item $q=n'\times k$ candidates $X_i=\{x_i^1,\cdots, x_i^q\}$ in $X$;
    \item $r=n'^2\times k^2$ candidates $Y_i=\{y_i^1,\cdots,y_i^r\}$ in $Y$.
\end{itemize}
Let~$N[c_i^1]$ denote the set of candidates in~$C_1$, whose corresponding vertices are in~$N[v_i]$.
Hereby, we define three arbitrary but fixed complete orders of $C_1$, $X$, and $Y$, respectively, denoted as $\overrightarrow{C_1}$, $\overrightarrow{X}$, and $\overrightarrow{Y}$.
For two sets $S'$ and $S$ with $S'\subset S$ and an order $\overrightarrow{S}$ of $S$, $\overrightarrow{S'}$ denotes the order resulted by deleting the elements in $S\setminus{S'}$ from $\overrightarrow{S}$.
We use $\overleftarrow{S}$ to denote the contrary order of $\overrightarrow{S}$.
We construct a vote $v_i^1$ in $A$ for each vertex $v_i\in  \mathcal{V}$:\\
\indent\indent
 $v_i^1:\; \overrightarrow{N[c^1_i]}>\overrightarrow{Y_i}>c_w>c_h>p>\overrightarrow{X_i}>
\overrightarrow{C_1\setminus N[c^1_i]}>\overrightarrow{X\setminus X_i}>\overrightarrow{Y\setminus Y_i}$.\\
Thus, $|A|=n'$ and $|C|=n'+n'^2\times k+n'^3k^2$.
Since the  candidate $c_i^1$ is ranked in front of $p$ only in the votes corresponding to the vertices in $N[v_{i}]$, there are exactly~$d_i=|N[c^1_i]|$ votes~$v\in A$ satisfying~${\tt diff}(c^1_i,p,\{v\})\geq r+2$.
Recall ${\tt diff}(c,c',\{v\})={\tt score}(c,v)-{\tt score}(c',v)$, where ${\tt score}(c,v)$ denotes the Borda-score that $c$ receives from $v$.
The other $n'-d_i$ votes $v$, where $c_i^1$ is ranked behind $p$, satisfy ${\tt diff}(c_i^1,p,\{v\})\geq -(q+n')$.
Thus, we have:
\begin{equation*}
\begin{split}
     {\tt diff}(c_i^1,p,A)&\geq d_i(r+2)-(n'-d_i)(q+n')\\
     &=d_i(n'^2k^2+2)-(n'-d_i)(n'k+n')\\
    &\geq d_in'^2k^2-n'(n'k+n')+2d_i\\
    &\geq 2n'^2k^2-2n'^2k+4\geq 4.
\end{split}
\end{equation*}
We set~$s_i={\tt diff}(c^1_i,p,A)-4$.
It is clear that $s_i\geq 0$.

Now, we construct the CCDV-Borda instance $(E=(C,V),\ell)$ from the auxillary election.
For each vertex $v_i\in \mathcal{V}$, we construct $s_i$ candidates in $Z_i$, i.e., ~$Z_i=\{z_i^1,\ldots, z_i^{s_i}\}$, and set~$Z=\bigcup_{1\le i\le n'}Z_i$.
Let~$C=C'\cup Z\cup \{c_{a},c_{b}\}$.
We define an arbitrary but fixed order of $Z$, denoted as $\overrightarrow{Z}$.
Finally, the vote set~$V$ consists of three subsets:~$V=V_1\cup V_2\cup V_3$, where~$V_1$ contains~$n'$ votes and is constructed in a similar way as the set~$A$.
For each~$v_i\in \mathcal{V}$, we create one vote $v_i^1$ in~$V_1$:\\
\indent
$v_i^1:\; \overrightarrow{N[c^1_i]}>\overrightarrow{Y_i}>c_w>c_h>p>\overrightarrow{X_i}>\overrightarrow{C_1\setminus N[c^1_i]}>\overrightarrow{Z}>\overrightarrow{X\setminus X_i}>\overrightarrow{Y\setminus Y_i}>c_b>c_a$,\\
Then, for each $v_i\in \mathcal{V}$, we create one vote $v_i^2$ in $V_2$ and one vote $v_i^3$ in $V_3$:\\
\indent
$v_i^2: \overrightarrow{C_1\setminus \{c^1_i\}}>c_b>c_a>p>c_w>\overrightarrow{Z_i}>c_i^1>\overrightarrow{Z\setminus Z_i}>\overrightarrow{X}
>\overrightarrow{Y}>c_h$,\\
\indent
$v_i^3:\;p>c_w>c^1_i>\overleftarrow{C_1\setminus \{c^1_i\}}>c_b>c_a>\overleftarrow{Z}>\overleftarrow{X}>\overleftarrow{Y}>c_h$.\\
We set $\ell=k$.
Observe that for each $c_i^1\in C_1$, we have ${\tt diff}(c_i^1,p,\{v_i^2,v_i^3\})=0$, if $i\neq j$; otherwise, ${\tt diff}(c_i^1,p,\{v_i^2,v_i^3\})=-s_i-4$.
Therefore, we have that for each $c_i^1\in C_1$, ${\tt diff}(c_i^1,p,V_1)=s_i+4$, ${\tt diff}(c_i^1,p,V_2\cup V_3)=-s_i-4$, and ${\tt diff}(c_i^1,p,V)=0$.
This implies that $p$ is not the unique winner.

Before proving the correctness of the construction, we analyse the possible final score of each candidate after deleting at most $\ell$ votes.
Let $V^{**}$ be an arbitrary subset of $V$ with $|V^{**}|\leq \ell$.\\
\indent
1) The candidates $u\in X\cup Z$ are ranked behind $p$ in each vote. Then, ${\tt diff}(u,p,\{v\})\\
<0$ for each $v\in V$ and cannot be the winner by deleting at most $\ell$ votes.\\
\indent
2) There are $n'$ votes $v\in V_2$ where ${\tt diff}(c_{b},p,\{v\})=2$, and ${\tt diff}(c_a,p,\{v\})=1$, and $2n'$ votes $v\in V_1\cup V_3$ where ${\tt diff}(p,c_{a}, \{v\})>n'$, and ${\tt diff}(p,c_{b},\{v\})>n'$.
Thus, it holds ${\tt diff}(p,c_{b}, V\setminus{V^{**}})>(2n'-\ell)\times n'-n'\times 2\geq 0$ and ${\tt diff}(p,c_{a},V\setminus{V^{**}})>(2n'-\ell)\times n'-n'\times 1\geq 0$.
Candidates $c_{a}$, $c_{b}$ cannot be the winner by deleting at most $\ell$ votes.\\
\indent
3) There are $n'$ votes $v \in V_1$ where ${\tt diff}(c_h,p,\{v\})=1$ and $2n'$ votes $v\in V_2\cup V_3$ where ${\tt diff}(p,c_h,\{v\})>1$.
Thus, ${\tt diff}(p,c_h,V\setminus{V^{**}})>(2n'-\ell)\times 1-n'\times 1\geq 0$, and $c_h$ cannot be the winner by deleting at most $\ell$ votes.\\
\indent
4) For each candidate $y\in Y$, there is exactly one vote $v$ in $V_1$, with $r+2\geq {\tt diff}(y,p,\{v\})>0$, while all other votes $v'$ in $V_1$ satisfy ${\tt diff}(p,y,\{v'\})>0$.
Moreover, ${\tt diff}(p,y,\{v_i^2,v_i^3\})>r+2q$ for each $i$. 
Thus, ${\tt diff}(p,y,V\setminus{V^{**}})>(n'-\ell)(r+2q)-(r+2)\geq 0$, and $y$ cannot be the unique winner by deleting at most $\ell$ votes.\\
\indent
Therefore, only the candidates in $C_1\cup \{c_w, p\}$ can be the unique winner of the election by deleting at most $\ell$ votes.
Note that for each $1\leq i\leq n'$, we have ${\tt diff}(p,c_w,\{v_i^1,\\
v_i^2,v_i^3\})=0$ and ${\tt diff}(p,c_w,V)=0$.
Thus, all these possible winners have the same Borda-score.
Next, we show the equivalence between the instances~$(\mathcal{G}=(\mathcal{V},\mathcal{E}), k)$ and $(E=(C,V),\ell)$.

\indent
``$\Longrightarrow$'': Given a size-$\leq k$ dominating set~$DS$, we set~$V^{**}=\{v_i^1\in V_1 |\ v_i\in DS\}$.
As argued above, we only consider the scores of the candidates in~$C_1\cup \{c_w,p\}$. Since~$V^{**}\subseteq V_1$, it holds ${\tt diff}(c_w,p,\{v\})>0$ for each $v\in V^{**}$.
We have that before deleting votes,~${\tt diff}(c_w,p,V)=0$.
Thus, we have~${\tt diff}(p,c_w,V\setminus V^{**})>0$.
Let~$c_i^1\in C_1$ be the candidate corresponding to an arbitrary vertex~$v_i\in \mathcal{V}$. Since~$DS$ is a dominating set, there is a vertex~$v_j\in DS$ with~$v_i\in N[v_j]$. 
Therefore, there is a vote~$v_j^1\in V^{**}$ with ${\tt diff}(c_i^1,p,\{v_j^1\})>|Y_j|=r>(\ell-1)q$.
Since ${\tt diff}(p,c^1_i,\{v\})\leq |X|+|C_1|<q$ for each $v\in V_1$, we conclude~${\tt diff}(p,c_i^1,V\setminus V^{**})>0$ for all~$c_i^1\in C_1$. Thus,~$V^{**}$ is a solution of the CCDV-Borda instance.

\indent
``$\Longleftarrow$'': Suppose that there is a set~$V^{**}$ of at most~$\ell$ votes satisfying ${\tt diff}(p,c,V\setminus V^{**})>0$
for all~$c\in C\setminus{\{p\}}$. We consider in particular the candidates in~$C_1\cup \{c_w\}$. For~$c_w$, we have ${\tt diff}(c_w,p,V)=0$ and~${\tt diff}(c_w,p,\{v\})>0$ only for~$v\in V_1$.
This means that there is at least one vote~$v_i^1\in V_1$ in $V^{**}$.
If a vote $v_i^1\in V_1$ is deleted, then for all $u\in (C_1\setminus{N[c_i^1]})$, we have ${\tt diff}(u,p,V\setminus{\{v_i^1\}})-{\tt diff}(u,p,V)={\tt diff}(u,p,\{v_i^1\})>|X_i|=n'\times k$. 
There must be some other votes deleted to decrease the score of $u$.
Note that for all votes $v\in V_3$, ${\tt diff}(p,u,\{v\})>0$ and for all votes $v\in V_2$, ${\tt diff}(u,p,\{v\})\leq n'+1$.
Due to $(n'+1)(k-1)<n'\times k$, there must be a vote in $V_1$ deleted to get ${\tt diff}(u,p,V\setminus V^{**})<0$.
It means there must be a vote $v_j^1\in V_1$ in $V^{**}$ and $u\in N[c_j^1]$.
More preciously, $\forall c_i^1\in C_1, \exists v_j^1\in V^{**}, c_i^1\in N[c_j^1]$.
Let $\mathcal{V''}$ be the set of the vertices corresponding to the votes in $V^{**}$.
Then, $\forall v_i\in \mathcal{V}, \exists v_j\in V'', v_i\in N[v_j]$.
So, $\mathcal{V''}$ is a dominating set of $\mathcal{G}$ with at most $k$ vertices.
\end{proof}

For CCAC-Borda and CCDC-Borda, we can give similar reductions to prove their hardness results.
The key idea of the reductions is to guarantee that, before adding or deleting candidates, the scores of some candidates (denoted as $C^+$) are equal to $p$'s score.
To make $p$ become the unique winner, we have to increase $p$'s score or decrease the scores of the candidates $c$ in $C^+$.
To do this, for CCAC-Borda, we add a new candidate $c'$ between $p$ and $c$ in at least one vote.
It means that a vote $v_1: \cdots p>\cdots >c\cdots$ is transformed into $v_1': \cdots p>\cdots >c'>\cdots >c\cdots$.
Similarly, for CCDC-Borda, we delete a candidate between $c$ and $p$ in at least one vote. 
This means that a vote $v_1: \cdots c>\cdots>c'>\cdots >p\cdots$ is transformed into $v_1': \cdots c>\cdots >p\cdots$.

\begin{theorem}
\label{Complete-CCAC-ell}
CCAC-Borda problem is W[2]-hard with respect to the number of candidate additions $\ell$.
\end{theorem}
\begin{proof}
We prove this theorem by a reduction from {\bf Dominating Set} problem.
Given an instance of Dominating Set problem $(\mathcal{G}=(\mathcal{V}, \mathcal{E}), k')$ where $|\mathcal{V}|=n'$.
Similar to Theorem~\ref{Complete-CCDV-ell}, let $N[v_i']$ be the set of vertices incident to $v_i'$ including $v_i'$.
Firstly, we construct an instance $(E=(C,C^*,V),\ell)$ of CCAC-Borda problem by $(\mathcal{G}=(\mathcal{V}, \mathcal{E}), k')$.
For each vertex $v_i'\in \mathcal{V}$, we construct a candidate $c_i^1$ in $C_1$, $C_1=\bigcup_{1\leq i\leq n'}\{c_i^1\}$, and a candidate $c_i^2$ in $C_2$, $C_2=\bigcup_{1\leq i \leq n'}\{c_i^2\}$.
Let $C=C_2\cup {p}$, $C^*=C_1$.
For each vertex $v_i'$, we also construct two votes:

\begin{equation*}
\begin{split}
&v_i^1: \overrightarrow{C_2\setminus{c_i^2}}>p>\overrightarrow{N[c_i^1]}>c_i^2>\overrightarrow{C_1\setminus{c_i^1}},\\
&v_i^2: c_i^2>p>\overleftarrow{C_2\setminus{c_i^2}}>\overleftarrow{N[c_i^1]}>\overleftarrow{C_1\setminus{N[c_i^1]}}.
    \end{split}
\end{equation*}
Let $V_1=\bigcup_{1\leq i\leq n'}\{v_i^1\}$, $V_2=\bigcup_{1\leq i\leq n'}\{v_i^2\}$, $V=V_1\cup V_2$, and $\ell=k'$.

Before adding the candidates in $C_1$, each vote ranks the candidates in $C_2\cup \{p\}$.
The two votes $v_i^1, v_i^2$ constructed according to the vertex $v_i'$ are $\overrightarrow{C_2\setminus{c_i^2}}>p>c_i^2$ and $c_i^2>p>\overrightarrow{C_2\setminus{c_i^2}}$ respectively.
The score of $p$ is equal to each candidate $c_i^2$, that is ${\tt diff}(c_i^2,p,V)=0, \forall c_i^2\in C_2$.
$p$ is not the unique winner.

Furthermore, we can get that the position of $c_i^1$ is always behind $p$ in each vote, ${\tt pos}(c_i^1,v)-{\tt pos}(p,v)>0, \forall c_i^1\in C_1, \forall v\in V$.
So, it always hold ${\tt diff}(c_i^1,p,V)<0$, the added candidates $c_i^1$ cannot be the winner.
And, the candidates of $C_1$ are always behind $p$ and $c_i^2$ in $V_2$.
The adding candidates operations do not change the value of ${\tt diff}(c_i^2,p,V_2)$.
So, we just need to focus on the values of ${\tt diff}(c_i^2,p,V_1)$ after adding the candidates in $C_1$.
In the following, we show the equivalence of the two instances.

\noindent
``$\Longrightarrow$'': Suppose that there is a dominating set $DS$ in $\mathcal{G}$ with $|DS|\leq k'$.
Let $C^{**}$ be the set of candidates in $C_1$ corresponding to the vertices in $DS$.
It is clear $|C^{**}|=|DS|\leq k'=\ell$.
Since $DS$ is a dominating set of $\mathcal{G}$, it always holds $\forall v_i'\in \mathcal{V}, \exists v_j'\in DS, v_j'\in N[v_i']$.
The candidates in $C_1$ and $C_2$ correspond to the vertices in $\mathcal{V}$.
More preciously, candidates $c_i^1$ and $c_i^2$ correspond to the vertex $v_i'$.
So, it must be $\forall c_i^2\in C_2, \exists c_j^1\in C_1, c_j^1\in N[c_i^1]$.
It means, for each vote in $V_1$, we add at least a candidate between $c_i^2$ and $p$, and the score of each $c_i^2$ will be decreased by at least $1$ corresponding to p.
In consequence, it satisfies ${\tt diff}(p,c_i^2,V)>0, \forall c_i^2\in C_2$.
Since the added candidates cannot be the winner, $p$ is the unique winner.

\noindent
``$\Longleftarrow$'': Suppose that there is a set of candidates $C^{**}$ with $C^{**}\subset C_1$ and $|C^{**}|\leq \ell$, after adding the candidates in $C^{**}$, the candidate $p$ is the unique.
Let $\mathcal{V''}$ be the set of vertices corresponding to the candidates in $C^{**}$.
Since $p$ is the unique winner after adding candidates, it must be ${\tt diff}(p,c_i^2,V)>0, \forall c_i^2\in C_2$.
It means the score of $c_i^2$ is decreased by at least 1 corresponding to $p$ after adding candidates.
According to the structure of $\overrightarrow{C_2\setminus{c_i^2}}>p>\overrightarrow{N[c_i^1]}>c_i^2$ in $V_1$, only adding the candidates in $N[c_i^1]$ can decrease the score of $c_i^2$ corresponding to $p$.
So, for each $c_i^2$, there must be a candidate of $N[c_i^1]$ in $C^{**}$.
That is $\forall c_i^2 \in C_2, \exists c_j^1\in C^{**}, c_j^1\in N[c_i^1]$.
Since the candidates in $C_1$ are corresponding to the vertices in $\mathcal{V}$.
It always holds $\forall v_i'\in \mathcal{V}, \exists v_j'\in \mathcal{V''}, v_j'\in N[v_i']$.
Therefore, $\mathcal{V''}$ is a size-$\leq k'$ dominating set of $\mathcal{G}$.
\end{proof}

\begin{theorem}
\label{Complete-CCDC-ell}
CCDC-Borda problem is W[2]-hard with respect to the number of candidate deletions $\ell$.
\end{theorem}
\begin{proof}
The proof is similar to the proof in Theorem~\ref{Complete-CCAC-ell} by a reduction from {\bf Dominating set} problem.
Given an instance of dominating set problem $(\mathcal{G}=(\mathcal{V},\mathcal{E}), k')$ where $|\mathcal{V}|=n'$.
Let $N[v_i']$ be the set of vertices incident to $v_i'$ including $v_i'$.
Firstly, we construct an instance of CCDC-Borda problem according to $(\mathcal{G}=(\mathcal{V},\mathcal{E}), k')$.
For each vertex $v_i'\in \mathcal{V}$, we construct a candidate $c_i^1$ in $C_1$, $C_1=\bigcup_{1\leq i\leq n'}\{c_i^1\}$; a candidate $c_i^2$ in $C_2$, $C_2=\bigcup_{1\leq i\leq n'}\{c_i^2\}$, and $d_i=|N[v_i']|$ dummy candidates $x_i^j$ in $X_i$, $X=\bigcup_{1\leq i\leq n'}X_i$.
The construction of candidates in $X_i$ aims at making the score of each $C_2$-candidate equals to $p$.
Let $C=C_1\cup C_2 \cup X\cup \{p\}$.
For each vertex $v_i'\in \mathcal{V'}$, we construct two votes:
\begin{equation*}
    \begin{split}
        &v_i^1: c_i^2>\overrightarrow{N[c_i^1]}>p>\overrightarrow{C_2\setminus{\{c_i^2\}}}>\overrightarrow{X}>\overrightarrow{C_1\setminus{N[c_i^1]}},\\
    &v_i^2: \overleftarrow{C_2\setminus{\{c_i^2\}}}>p>\overrightarrow{X}>c_i^2>\overrightarrow{X\setminus{X_i}}>\overrightarrow{C_1}.
    \end{split}
\end{equation*}
Let $V_1=\bigcup_{1\leq i\leq n'}\{v_i^1\}$, $V_2=\bigcup_{1\leq i\leq n'}\{v_i^2\}$, $V=V_1\cup V_2$, $\ell=k'$.

Before deleting candidates, the candidates in $X$ are always behind of $p$ in each vote, ${\tt pos}(x_i^j, v)>{\tt pos}(p,v), \forall x_i^j\in X, v\in V$.
It always holds ${\tt diff}(p,x_i^j,V)>0$.
The candidate $x_i^j$ cannot be the unique winner.
In addition, the position of $p$ is always forward of $c_j^1$ in $v_i^2\in V_2$, ${\tt pos}(p,v_i^2)<{\tt pos}(c_j^1,v_i^2)$, and ${\tt diff}(p,c_j^1,\{v_i^2\})>|X|$.
In $v_i^1$, even the position of $c_j^1$ can be forward of $p$, it still holds ${\tt diff}(c_j^1,p,\{v_i^1\})<\sum_{1\leq i\leq n'}|N[c_i^1]|=|X|$.
So, for each $C_1$ candidate $c_j^1$, ${\tt diff}(p,c_j^1,\{v_i^1\})+{\tt diff}(c_j^1,p,\\
\{v_i^2\})>0$ is always hold.
Since $|N[c_i^1]|=|X_i|=d_i$, there are $d_i$ candidate between $p$ and $c_i^2$ in $v_i^1$.
It means ${\tt diff}(c_i^2,p,\{v_i^1\})={\tt diff}(p,c_i^2,\{v_i^2\})=d_i+1$.
The candidates in $C_2\setminus{\{c_i^2\}}$ are always symmetrical corresponding to $p$ in $v_i^1, v_i^2$, since they are $p>\overrightarrow{C_2\setminus{c_i^2}}$ and $\overrightarrow{C_2\setminus{\{c_i^2\}}}>p$ in $v_i^1$ and $v_i^2$ respectively.
Therefore, the score of each $C_2$ candidate is equal to $p$'s, ${\tt diff}(p,c_i^2,V)=0, \forall c_i^2\in C_2$.
$p$ is not the unique winner.

Since the positions of candidates in $X$ are always behind of $p$ in $V$, deleting the candidates in $X$ is meaningless to make $p$ being the unique winner.
We suppose $C^{**}$ is a solution of $(E=(C,V),\ell)$.
If there is a candidate $x_i^j$ in $C^{**}$, then $C^{**}\setminus{\{x_i^j\}}$ must also be a solution of $(E=(C,V),a)$.
If there is a candidate $c_i^2$ in $C^{**}$, according to the structure of $V_1\cup V_2$, the position of $c_i^2$ is not between $p$ and $c_j^2$.
So, the aim of deleting $c_i^2$ is only to make ${\tt diff}(p,c_i^2,V)>0$.
If we delete candidate $c_i^1$ instead of $c_i^2$, it can also make sure ${\tt diff}(p,c_i^2,V)=1>0$.
It means $(C^{**}\cup \{c_i^1\})\setminus{\{c_i^2\}}$ is also a solution of $(E=(C,V),\ell)$.
In summary, each solution of $(E=(C,V),\ell)$ can be transformed in to another solution whose elements are all candidates in $C_1$.
So, in the following, we just consider the condition that deleting the candidates in $C_1$.

\noindent
``$\Longrightarrow$'': Suppose that there is a dominating set $DS$ in $\mathcal{G}$ with $|DS|\leq k'$.
Let $C^{**}$ be the set candidates corresponding to $DS$ with $C^{**}\subset C_1$ and $|C^{**}|=|DS|\leq k'=\ell$.
Since $DS$ is a dominating set of $\mathcal{G}$, it must hold $\forall v_i'\in \mathcal{V}, \exists v_j'\in DS, v_j'\in N[v_i']$.
And the candidates in $C_1$ or $C_2$ are corresponding to the vertices in $\mathcal{V'}$.
It always holds $\forall c_i^2\in C_2, \exists c_j^1\in C^{**}, c_j^1\in N[c_i^2]$.
It means, for each candidate $c_i^2$, there must be a candidate in $N[c_i^1]$ deleted.
So, in each $v_i^1$, there is at least one candidate between $c_i^2$ and $p$ deleted.
The score of each $c_i^2$ is decreased by at least 1 corresponding to $p$, ${\tt diff}(p,c_i^2,V)\geq 1, \forall c_i^2\in C_2$.
According to the above analyse, the candidates in $C_1\cup X$ cannot be the winner.
So, $p$ is the unique winner.

\noindent
``$\Longleftarrow$'': Suppose that there is a set of candidates $C^{**}$ with $|C^{**}|\leq \ell$ that deleting the candidates in $C^{**}$ can make $p$ the unique winner of $(E=(C,V),\ell)$.
Let $\mathcal{V''}$ be the set of vertices corresponding to the candidates in $C^{**}$, $|\mathcal{V''}|=|C^{**}|\leq \ell =k'$.
Because $p$ is the unique winner after deleting candidates, it must be ${\tt diff}(p,c_i^2,V)>0, \forall c_i^2\in C_2$.
Since before deleting candidates, the score of each $C_2$ candidate $c_i^2$ is equal to $p$'s, ${\tt diff}(p,c_i^2,V)=0$.
So, deleting the candidates in $C^{**}$ can decrease the score of each $C_i^2$ by at least 1 corresponding to $p$.
According the structure of $c_i^2>\overrightarrow{N[c_i^1]}>p$ in $v_i^1$, we can only delete the candidates in $N[c_i^1]$.
So, for each candidate $c_i^2$, there must be a candidate $c_j^1\in N[c_i^2]$ in $C^{**}$.
It means $\forall c_i^2\in C_2, \exists c_j^1\in C^{**}, c_j^1\in N[c_i^2]$.
The candidates in $C_1$ and $C_2$ are corresponding to the vertices in $\mathcal{V}$.
It can be transformed into $\forall v_i'\in \mathcal{V}, \exists v_j'\in \mathcal{V''}, v_j'\in N[v_i']$.
$\mathcal{V''}$ is a size-$\leq k'$ dominating set of $\mathcal{G}$.
\end{proof}

\section{Complexity for Top-Truncated Votes}
In this section, we consider the complexity of Borda control problems with $t$-truncated votes.
Hemaspaandra and Schnoor~\cite{HS-ECAI-2016} proved that $t$-CCAV-$\delta$ and $t$-CCDV-$\delta$ are NP-hard for every $t\geq 3$ and polynomial-time solvable for $t\leq 2$.
We present FPT algorithms for $t$-CCDV-$\delta$ and $t$-CCAV-$\delta$ with $t$ and the number of operations $\ell$ as parameter (Theorems~\ref{Top-CCDV-ell,t},~\ref{Top-CCAV-ell,t}).
For candidate control problems, we show the NP-hardness and W-hardness results for $t$-CCAC-$\delta$ and $t$-CCDC-$\delta$ with $t\geq 2$ (Theorem~\ref{top-CCDC-up-w}-\ref{top-CCAC-CCDC-av-w}).

\begin{theorem}
\label{Top-CCDV-ell,t}
$t$-CCDV-$\delta$ with $\delta \in \{{\tt Borda}_{\uparrow}, {\tt Borda}_{\downarrow}, {\tt Borda}_{av}\\\}$ is FPT with respect to the number of deletions $\ell$ and the maximal number $t$ of ranked candidates in each vote.
\end{theorem}
\begin{proof}
Let~$(E=(C, V), \ell)$ be a~$t$-CCDV-$\delta$ instance with~$n=|V|$ and~$m=|C|$.
Our algorithm consists of the following 5 steps
:\\
\indent
(1) Calculate the score of each candidate $c \in C$ by the rule of $\delta \in \{{\tt Borda}_{\uparrow}, {\tt Borda}_{\downarrow}, \\
{\tt Borda_{av}}\}$;\\
\indent
(2) Classify each vote by $|v|$ and ${\tt pos}(v, p)$;\\
\indent
(3) Enumerate all feasible vote type combinations in the solution and calculate the final score $F(p)$ of $p$ for each vote type combination;\\
\indent
(4) For each possible vote type combination, remove redundant votes according to the combination and the corresponding $F(p)$;\\
\indent
(5) Check all subsets $V''$ of the set of the remaining votes $V'$ of size at most $\ell$,  whether deleting all votes in $V''$ can make $p$ be the unique winner of the election.

In the first step of the algorithm, we calculate the score of each candidate $c\in C$, that is, ${\tt score}(c,V)$.
To simplify the presentation, we adopt a slightly different definition of the Borda-score in the ${\tt Borda}_{av}$ case, that is, ${\tt score}(c,V)=\sum_{v\in V}({\tt score}(c,v)-\frac{m-|v|-1}{2})$, where ${\tt score}(c,v)$ is the score of $c$ receiving from vote $v$ according to the definition of Borda$_{av}$ in Section 2.
Note that, the new definition does not change the value of ${\tt diff}(c,c',V)$ for any two candidates $c$ and $c'$.
Then, we classify votes $v$ into different types according to the position of $p$ in $v$ and the length of $v$: A vote $v$ is of type $ty_i^j$ with $0\leq i={\tt pos}(v,p)\leq |v|$, $0\leq j=|v|\leq t$.
There are totally at most $\frac{(t+1)(t+2)}{2}$ different types.\\
\indent
Then,
we enumerate all ``feasible vote type combinations''.
A feasible vote type combination $T$ is a collection of pairs $(ty_i^j,\ell_i^j)$ with $0\leq j\leq t$, $0\leq i\leq j$, $0\leq \ell_i^j\leq \ell$, and $\sum_{i=0}^t\sum_{j=0}^t\ell_i^j\leq \ell$.
A pair $(ty_i^j,\ell_i^j)$ means that $\ell_i^j$ votes of type $ty_i^j$ are deleted.
For each feasible vote type combination, we compute a ``final'' score $F(p)$ for $p$ based on ${\tt score}(p,V)$.
Note that ${\tt score}(p,v)$ for all votes $v$ of type $ty_i^j$ is the same and easy to compute without knowing the exact order of $v$: if $i=0$, then ${\tt score}(p,v)=0$ $({\tt Borda}_{\uparrow}, {\tt Borda}_{\downarrow}, {\tt Borda}_{av})$; otherwise, ${\tt score}(p,v)=m-i$ $({\tt Borda_{\uparrow}}, {\tt Borda}_{av})$ or $j-i+1$ (Borda$_\downarrow$).
Then, set ${\tt score}(p, ty_i^j):={\tt score}(p,v)$ and $F(p):={\tt score}(p,V)-\sum {\tt score}(p, ty_i^j)$.
Next, we divide the candidates in $C\setminus{\{p\}}$ into two subsets: $C^+=\{c\in C\ |\ {\tt score}(c,V)\geq F(p)\ {\tt and}\ c\neq p\}$ and $C^-=C\setminus{(C^+\cup \{p\})}$,
and remove some redundant votes with respect to the position of candidates in $C^+$.
We call two votes $v$ and $v'$ are similar, if they have the same length and $\forall c\in (C^+\cup \{p\})$, ${\tt pos}(v,c)={\tt pos}(v',c)$.
The elimination rule is as follows:
\begin{itemize}
    \item {\bf [DR]}: If there are more than~$\ell+1$ pairwisely similar votes in~$V$, then we remove them until~$\ell+1$ many remain.
\end{itemize}
\indent
Let $V'$ be the set of votes remaining after applying DR exhaustively.
Finally, we iterate over all subsets $V''$ of $V'$ with $|V''|\leq \ell$.
For each $V''$, we examine whether $p$ is the unique winner of $E=(C,V\setminus{V''})$.
The pseudo-code of the algorithm is shown in Algorithm~\ref{alg:algorithm-CCDV}.

\begin{algorithm}[tb]
\caption{FPT algorithm for $t$-CCDV-$\delta$}
\label{alg:algorithm-CCDV}
\raggedright
\textbf{Input}: An instance of $t$-CCDV-$\delta$: $(E=(C,V),\ell)$\\
\textbf{Output}: A solution $V''$ or no such solution

\begin{algorithmic}[1] 
\FOR{each candidate $c\in C$}
\STATE Calculate ${\tt score}(c,V)$
\ENDFOR
\FOR{each vote $v\in V$}
\STATE Compute $ty_{{\tt pos}(p,v)}^{|v|}$
\ENDFOR
\FOR{each combination of vote types}
\STATE Calculate $F(p)$ and classify $C$ into $C^+$ and $C^-$
\IF{$|C^+|>t*\ell$}
\STATE Go to the next for-iteration
\ENDIF
\STATE Construct $V'$ by eliminating some votes from $V$ according to DR,
\FOR{each subset $V''$ of $V'$}
\IF{$|V''|\leq \ell$}
\IF{$p$ is the unique winner of $E=(C,V\setminus{V''})$}
\STATE Return $V''$
\ENDIF
\ENDIF
\ENDFOR
\ENDFOR
\end{algorithmic}
\end{algorithm}

{\bf Claim.} The above algorithm runs in $O^*(2^{t\cdot\ell^2+\ell}\cdot t^{t\cdot\ell+\ell}\cdot(\ell+1)^{\frac{(t+1)(t+2)}{2}+\ell})$ time.\\
\indent
{\bf Proof of claim.}
In order to make~$p$ win the election, we have to make all candidates $c\in C^+$ satisfying ${\tt score}(c,V\setminus{V''})<F(p)$ where $V''$ is the set of votes deleted.
The first step is clearly doable in polynomial time.
The number of all feasible vote type combinations and thus possible $F(p)$ values is at most $(\ell+1)^{\frac{(t+1)(t+2)}{2}}$, since at most $\ell$ votes of each type can be deleted.
Next, for each possible $F(p)$,~$|C^+|\le t\cdot\ell$, since we are allowed to delete at most~$\ell$ $t$-truncated votes and deleting one vote $v$ can decrease the scores of at most $t$ candidates.
Two votes~$v$ and~$v'$ are not similar, if ${\tt pos}(v, c)\ne {\tt pos}(v',c)$ for one~$c\in (C^+\cup\{p\})$ or $|v|\neq |v'|$.
Therefore, there are less than~$t\cdot 2^{|C^+|+1}\cdot t!<2^{|C^+|+1}\cdot t^{t+1}$ pairwisely non-similar votes and each of them can have at most~$\ell+1$ similar votes in~$V'$ by DR,
which leads to less than~$2^{t\cdot \ell+1}\cdot t^{t+1}\cdot(\ell+1)$ remaining votes. Then, examining all size-$\le \ell$
subsets of the remaining votes can be done in~$O^*(2^{t\cdot\ell^2+\ell}\cdot t^{t\cdot\ell+\ell}\cdot(\ell+1)^\ell)$ time. Therefore, the total running time is bounded by~$O^*(2^{t\cdot\ell^2+\ell}\cdot t^{t\cdot\ell+\ell}\cdot(\ell+1)^{\frac{(t+1)(t+2)}{2}+\ell})$.

Next, we show the correctness of the algorithm.
It is clear that every output generated by our method
is a solution of the~$t$-CCDV-$\delta$ instance.
For the reversed direction, note that every solution~$V''$ for a~$t$-CCDV-$\delta$ instance conforms to one of the vote type combinations, since each vote in~$V''$ is of type~$ty_i^j$ for some~$0\le i\leq j \le |v|\leq t$, and the number of the votes in~$V''$ is bounded by~$\ell$. Let~$X$ denote this combination.
Clearly,~$X$ has
been enumerated by the
algorithm and the final score~$F(p)$ computed by the algorithm for~$X$
is then equal to~${\tt score}(p, V\setminus V'')$.
Further, note that the algorithm removes votes from~$V$ only by applying DR to similar votes. Thus, while
the algorithm considers~$X$, if no vote
in~$V''$ has been removed from~$V$ by DR, then the algorithm can
find~$V''$ by the brute-force enumeration of size-$\le \ell$
subsets. Suppose that there is at least one vote~$v\in V''$ removed by an
application of DR. This occurs, only if there are~$\ell+1$ votes
similar to~$v$, which remain in~$V$ after the application of
DR. Let~$R$ denote the set of these~$\ell+1$ votes. Clearly,~$R\setminus
V''\ne \emptyset$.
According to the definition of similar votes, for all~$v'\in R$ and all~$c\in (C^+\cup \{p\})$, we have~${\tt pos}(v,c)={\tt pos}(v', c)$ and $|v|=|v'|$, implying that~${\tt score}(c,V\setminus V'')={\tt score}(c, V\setminus (V''\setminus \{v\}\cup \{v'\}))$ for each vote~$v'\in (R\setminus V'')$ and each candidate~$c\in (C^+\cup\{p\})$.
Recall that~$C^+$ is the set of ``critical'' candidates for this combination~$X$, that is, the candidates whose scores are at least $F(p)$.
Since the candidates~$c\in C^-$ have~${\tt score}(c, V)<F(p)={\tt score}(p, V\setminus{V''})$ and vote deletions can only decrease the scores of candidates, we can conclude~${\tt score}(c,V\setminus V'')<F(p)$ for all $c\in C^-$.
By replacing~$v$ by~$v'$, we receive another solution for the~$t$-CCDV-$\delta$ instance. Repeating the above replacement, we construct a solution, whose votes are all contained in~$V'$ after DR has been exhaustively applied.
Thus, the brute-force enumeration of size-$\le \ell$ subsets can find this solution. This completes the proof of the correctness of the algorithm.
\end{proof}

The algorithm solving $t$-CCAV-$\delta$ shares some common features with the one for $t$-CCDV-$\delta$.
It also enumerates all feasible vote type combinations.
However, we have here only~$\frac{t(t+1)}{2}$ types $ty_i^j$, with $1\leq i \leq j\leq |v|\leq t$, since all added votes should contain~$p$.
The main difference lies in the candidates in $C^+$, since adding votes to increase the score of~$p$ might increase the scores of some candidates in~$C\setminus{\{p\}}$ as well, but we cannot identify these candidates in advance.
Thus, the concept of $C^+$ does not work here.
To overcome this difficulty, we apply a series of data reduction rules to reduce the ``redundant'' votes in unregistered vote set~$V^*$, which is the set of votes that can be added to the election, such that~$|V^*|$ can be bounded by a function of~$t$ and~$\ell$.
Finally, an enumeration of size-$\le \ell$ subsets is also applied to~$V^*$.

\begin{theorem}
\label{Top-CCAV-ell,t}
$t$-CCAV-$\delta$ with $\delta \in \{{\tt Borda}_{\uparrow}, {\tt Borda}_{\downarrow}, {\tt Borda}_{av}\\\}$ is FPT with respect to the number of additions $\ell$ and the maximal number  $t$ of ranked candidates in each vote.
\end{theorem}

\begin{proof}
Let~$(E=(C, V, V^*), \ell)$ be an instance of $t$-CCAV-$\delta$.
In the first step of the algorithm, we compute the score of each candidate~$c\in C$ with respect to~$V$, that is, ${\tt score}(c, V)$.
Then, we classify votes into different types and enumerate all feasible vote type combinations.
The rule of classification is the same as the one for deleting votes.
For each of these combinations, we compute~$F(p)$ in a similar way.
Then, we apply the following two rules.
Here, two votes~$v$ and~$v'$ are identical, if for all~$c\in C$,
it holds~${\tt pos}(v,c)={\tt pos}(v',c)$.
\begin{itemize}
    \item {\bf [Cleaning Rule]}: If there exists a vote~$v\in V^*$ with~${\tt
  score}(c, v)\\ + {\tt score}(c, V)\ge F(p)$ for a candidate~$c\in C$, then remove~$v$ from~$V^*$.
  \item {\bf [Identical Votes Rule]}:  If there are more than~$\ell+1$ identical votes in~$V^*$, then we keep only~$\ell+1$ many such identical votes in~$V^*$.
\end{itemize}
Next, we partition the votes in~$V^*$ according to their types,
$V^*_i=\{v\in V^*\ |\ p=c(v, i)\}$.
We define the following upper-bounds:
\begin{itemize}
    \item $g_1=(t-1)( \ell-1)\ell$,
    \item $g_2=2(t-1)(\ell-1)g_1=2(t-1)^2(\ell-1)^2\ell$,\\
    $\vdots$
    \item $g_{t-1}=(t-1)(t-1)(\ell-1)g_{t-2}=(t-1)!(t-1)^{t-1}(\ell-1)^{t-1}\ell$.
\end{itemize}
Then, we perform the following operations to each~$V^*_i$.
\begin{itemize}
    \item {\bf [DR1]}: If there are more than $g_1+1$ votes in~$V^*_i$, which are all identical
except for one position~$1\le j\le t$ with~$i\ne j$, then keep~$g_1+1$ many of them in~$V^*_i$.
\item {\bf [DR2]}: If there are more than~$g_2+1$ votes in~$V^*_i$, which are all identical
except for two positions~$1\le j_1<j_2\le t$ with~$i\ne j_1$ and~$i\ne j_2$, then
keep~$g_2+1$ many of them in~$V^*_i$.\\
$\vdots$
\item {\bf [DR(t-1)]}: If there are more than~$g_{t-1}+1$ votes in~$V^*_i$, which are all identical
except for~$t-1$ positions~$1\le j_1<j_2\ldots<j_{t-1}\le t$ with~$i\notin\{j_1, \ldots, j_{t-1}\}$, then
keep~$g_{t-1}+1$ many of them in~$V^*_i$.
\end{itemize}
The votes removed by the above operations are called redundant votes.
Note that we apply first the cleaning and identical votes rules.
Then, DR1-DR$(t-1)$ are applied in the order of the presentation, that is, for
each~$1\le i\le t$ and each~$1<j\le t-1$, we apply DR$j$ to~$V^*_i$ only
if DR$(j-1)$ is not applicable to~$V^*_i$.
Finally, we iterate over all size-$\le \ell$ subsets of the remaining votes in~$\bigcup_{1\le i\le t}V^*_i$ and examine whether one of them is a solution for the given instance.
The pseudo-code of this algorithm is shown in Algorithm~\ref{alg:algorithm-CCAV}.

\begin{algorithm}[tb]
\caption{FPT algorithm for $t$-CCAV-$\delta$}
\label{alg:algorithm-CCAV}
\raggedright
\textbf{Input}: An instance of $t$-CCAV-$\delta$: $(E=(C,V,V^*), \ell)$\\
\textbf{Output}: A solution $V''$ or no such solution

\begin{algorithmic}[1] 
\FOR{each candidate $c\in C$}
\STATE Calculate ${\tt score}(c,V)$
\ENDFOR
\FOR{each vote $v\in V^*$}
\STATE Compute $ty^{|v|}_{{\tt pos}(p,v)}$
\STATE Add $v$ to $V_{pos(v,p)}^*$
\ENDFOR
\FOR{each combination of vote types}
\STATE Calculate $F(p)$
\FOR{$1\leq i\leq t$}
\STATE Apply Cleaning Rule and Identical Votes Rule to $V^*_i$
\STATE Apply ${\tt DR}_{i}$ to $V^*_i$
\ENDFOR
\STATE Set $V'=\bigcup_{i=1}^tV_i^*$
\FOR{each subset $V''$ of $V'$}
\IF{$|V''|\leq \ell$}
\IF{$p$ is the unique winner of $E=(C,V\setminus{V''})$}
\STATE Return $V''$
\ENDIF
\ENDIF
\ENDFOR
\ENDFOR
\end{algorithmic}
\end{algorithm}

{\bf Claim.} The algorithm runs in $O^*((t!)^{\ell} \cdot (t\cdot\ell)^{t\cdot\ell}\cdot(\ell+1)^{\frac{t(t+1)}{2}+\ell})$ time.\\
\indent
{\bf Proof of Claim.}
As in the proof of Theorem~\ref{Top-CCDV-ell,t}, there are at most~$(\ell+1)^{\frac{t(t+1)}{2}}$ vote type combinations. Then, for each combination, there
are~$t$ subsets~$V^*_i$ in the partition of~$V^*$ and by the above data reduction rules,
$|V^*_i|\le (g_1+1)+\cdots+(g_{t-1}+1)\leq (t-1)!\cdot(t-1)^{t}\cdot(\ell-1)^{t-1}\cdot\ell$ for each~$1\le i\le t$. Therefore, there are less than~$t!\cdot(t\cdot\ell)^{t}\cdot (\ell+1)$ votes in~$\bigcup_{1\le i\le t} V^*_i$.
Furthermore, the application of the cleaning and identical votes rules needs clearly polynomial time.
DR1-DR$(t-1)$ are all doable in $O(n\cdot t^t)$ time, since we can find in $O(t^{t})$ time all votes in $V^*$, that are identical to a given vote except for $i$ positions for all $1\leq i \leq t$.
Finally, the enumeration of size-$\le \ell$ subsets of $\bigcup_{1\leq i \leq t}V^*_i$ has to examine less than~$(t!)^{\ell}\cdot(t\cdot\ell)^{t\cdot\ell}\cdot(\ell+1)^{\ell}$ possibilities. We need polynomial time to verify
whether one of these subsets is a solution for the given
instance. Therefore, the total running time of the algorithm is
bounded by~$O^*((t!)^{\ell} \cdot (t\cdot\ell)^{t\cdot\ell}\cdot(\ell+1)^{\frac{t(t+1)}{2}+\ell})$.

The algorithm removes votes by applying the cleaning rule, the
identical votes rule, or DR1-DR$(t-1)$. Since the enumeration of
vote type combinations and size-$\le \ell$ subsets are obviously correct,
the correctness of the algorithm follows directly from the correctness
of these rules.

The cleaning rule is correct, since the vote removed by the rule
clearly cannot be added to~$V$. Adding it to~$V$ would result in that the
final score of a candidate~$c\in (C\setminus{\{p\}})$ exceeds~$F(p)$, the final score
of~$p$ with respect to the current combination of vote types. The
correctness of the identical rule is obvious, since we cannot add more
than~$\ell$ votes to~$V$.

Note that the algorithm iterates over the combinations of vote
types one-by-one and DR$r$'s are applied under the assumption that
the solution sought for conforms to the current combination. Therefore, we
consider in the following an arbitrary but fixed combination~$X$ of vote types.
We show the correctness of DR$r$ in the order from~$r=1$ to~$r=t-1$,
and by proving that every application of DR$r$ is correct with the following claim.

\medskip
\noindent
{\bf Claim.} Let~$W$ be the union of all $V^*_i$'s before the application of DR$r$
and~$W'$ be the union after the application. If there is a
solution for the given $t$-CCAV-$\delta$ instance which conforms to~$X$
and adds at most~$\ell$ votes in~$W$ to~$V$, then there is a solution
which conforms to~$X$ and adds at most~$\ell$ votes in~$W'$ to~$V$.

\medskip

\noindent
{\bf Proof of Claim.} Let~$Z\subseteq W$ with~$|Z|\le \ell$ be a solution for
the given instance conforming to~$X$. If~$Z\subseteq W'$, then we are done; otherwise,
let~$v\in (Z\setminus W')$. Let~$C(v)$ be the set of candidates, who
occur in the vote~$v$, and~$C^+=\cup_{u\in Z}C(u)$.
Clearly, $|C^+|\leq (t-1)\ell+1$ and $|C^+\setminus{C(v)}|\leq (\ell-1)(t-1)$.

First consider~$r=1$. Since~$v$ is removed from~$W$ by DR1, there is
a~$V^*_i\subseteq W$ containing more than~$g_1+1=(t-1)(\ell-1)\cdot\ell+1$ votes,
including~$v$, which are all identical except for one position~$1\le
j\le t$ with~$i\ne j$. Let~$Y$ be the set of these votes. Then,
$|\bigcup_{u\in Y}\{c(u,j)\}|>(t-1)(\ell-1)$, since otherwise, there would be more
than~$\ell+1$ identical votes in~$V^*_i$, contradicting to the fact that the
cleaning rule is not applicable to~$V^*_i$. Then, there is a
vote~$w\in W'$ that differs from~$v$ only in the $j$-th position and
has~$c(w,j)\notin (C^+\setminus C(v))$.
Thus, set~$Z'=Z\setminus \{v\}\cup \{w\}$. By
the cleaning rule, the score of~$c(w,j)$ in~$V\cup Z'$ is less
than~$F(p)$, that is, the final score of~$p$ in~$V\cup Z$ according to~$X$.
The scores of other candidates are not increased by the
replacement. Thus,~$p$ has the
highest score in~$V\cup Z'$. Repeating the above argument to other votes
not in~$W'$, we can get a solution~$Z'$ with~$Z'\subseteq W'$.

Assume that the claim is true for DR1, ..., DR$(r-1)$. Consider an
application of DR$r$. From the precondition of DR$r$, there is
a~$V^*_i$ containing more than~$g_{r}+1=r!(t-1)^r(\ell-1)^r\cdot\ell+1$ votes, which are all identical
except for~$r$ positions~$1\le j_1<j_2\ldots<j_r\le t$
with~$i\notin\{j_1, \ldots, j_r\}$. Let~$Y$ be the set containing
these votes. Again, consider the $j_1$-th position. By the same reason,
$|\bigcup_{u\in Y}\{c(u,j_1)\}|>(t-1)(\ell-1)$, since otherwise, there would be more
than~$(r-1)!(t-1)^{r-1}(\ell-1)^{r-1}\cdot\ell+1$ votes in~$W$, which satisfy the precondition
of DR$(r-1)$, contradicting to the order of applications.
Let~$Y_1$ contain the votes~$u\in Y$ with~$c(u,j_1)\notin (C^+\setminus C(v))$. Note that $|Y_1|\ge g_r-(t-1)(\ell-1)g_{r-1}$.
Then, consider the $j_2$-th position. Again, there are at least $|Y_1|-(t-1)(\ell-1)g_{r-1}\ge g_r-2(t-1)(\ell-1)g_{r-1}$
votes~$u \in Y_1$, with~$c(u,j_2)\notin C^+\setminus C(v)$.
Let~$Y_2$ denote the set of these votes.
In this way, we can conclude that there is $Y_r\subset Y_{t-1}$ with $|Y_r|\ge g_r-r(t-1)(\ell-1)g_{r-1}>0$.
This means that there is a vote~$w\in W'$, which is identical to~$v$
except for~$r$ positions~$1\le j_1<j_2\ldots<j_r\le t$. Moreover, the
candidates~$c(w,j_1), \ldots, c(w,j_r)$ are not in~$C^+\setminus C(v)$.
Compared to~$Z$, adding $Z'=Z\setminus
\{v\}\cup \{w\}$ to~$V$ does not increase the scores of the candidates
other than~$c(w,j_1), \dots, c(w,j_r)$, whose final scores are less
than~$F(p)$ by the cleaning rule. Thus, $Z'$ is another
solution to the given instance. Repeating the above argument
to other votes in~$Z\setminus W'$, we can construct another solution
satisfying the claim. This completes the proof of the
claim and the proof of the correctness.
\end{proof}

Next, we consider the complexity of the control by candidate operations with $t$-truncated votes.
Note that no unregistered/registered candidate can turned to be register/unregistered candidate by candidate deletions/additions.
When $t=1$, there are at most one ranked candidate in each vote.
For $1$-CCAC-$\delta$, candidate additions cannot decrease the score of any candidate.
Therefore, we can conclude that if $p$ is not the unique winner of the current election, it is impossible to make $p$ the unique winner by adding candidates.
For $1$-CCDC-$\delta$, the candidate deletions can only decrease the score of the deleted candidates.
It means that $p$ can become the unique winner of the election if and only if the number of the candidates, whose scores are at least $p$'s before deleting candidates, does not exceed the number of operations $\ell$.
The details of our algorithm are shown in Algorithm~\ref{alg:algorithm-CCDC}.

\begin{algorithm}[tb]
\caption{Polynomial-time algorithm for $1$-CCDC-$\delta$}
\label{alg:algorithm-CCDC}
\raggedright
\textbf{Input}: An instance of $1$-CCDC-$\delta$: $(E=(C,V), \ell)$, where $C=\{c_1, c_2, \cdots, c_m\}\cup \{p\}$, $V=\{v_1, v_2, \cdots, v_n\}$ and $\ell$ is a \\
\textbf{Output}: A subset $C'$ of $C$ satisfying $|C'|\leq \ell$ and $p$ is the unique winner of election $(C/C', V)$ or there is no solution for instance $(E=(C,V), \ell)$

\begin{algorithmic}[1] 
\FOR{each candidate $c\in C$}
\STATE Calculate ${\tt score}(c,V)$
\ENDFOR
\STATE Set $C^+=\emptyset$
\FOR{each candidate $c\in C/\{p\}$}
\IF{${\tt score}(c,V)\geq {\tt score}(p,V)$}
\STATE $C^+ = C^+\cup \{c\}$
\ENDIF
\ENDFOR
\IF{$|C^+|>\ell$}
\STATE there is no solution for $(E=(C,V), \ell)$
\ELSE
\STATE $C^+$ is the solution for $(E=(C,V), \ell)$
\ENDIF
\end{algorithmic}
\end{algorithm}

\begin{theorem}
\label{thm:1-p}
 $1$-CCAC-$\delta$ and $1$-CCDC-$\delta$ with $\delta \in \{{\tt Borda}_{\uparrow}, {\tt Borda}_{\downarrow}, {\tt Borda}_{av}\}$ can be solved in $O(n\cdot m)$ time.
\end{theorem}

Next, we consider the complexity of $t$-CCAC-$\delta$ and $t$-CCDC-$\delta$ with $t\geq 2$.
For $t$-CCAC-${\tt Borda}_{\uparrow}$, after we add a candidate $c'$ in $v$, the score of each candidate $c$ with ${\tt pos}(v,c)>{\tt pos}(v,c')$ will be decreased by one, while the scores of other candidates remain unchanged.
Similarly, for $t$-CCDC-${\tt Borda}_{\downarrow}$, after we delete a candidate $c'$ in $v$, the score of each candidate $c$ with ${\tt pos}(v,c)<{\tt pos}(v,c')$ will be decreased by one, while the scores of other candidates remain unchanged.

To show the score of each candidate after the control operations more clearly, we present the following two examples.

\vbox{}
{\bf Example 4:}Given an instance of $3$-CCAC-${\tt Borda}_{\uparrow}$ $(E=(C, C^*, V), \ell)$ where $C=\{c_1, p\}$, $C^*=\{c_2, c_3\}$, $V=\{v_1, v_2\}$, $\ell=1$, and $v_1: p>c_3>c_1$, $v_2: c_2>c_1>p$.
Before adding candidates, $v_1$ is $p>c_1$ and $v_2$ is $c_1>p$.
The score of $p$'s is equal to $c_1$'s, ${\tt score}(p, V)={\tt score}(c_1, p)=1$.
If we add the candidate of $c_2$, vote $v_2$ turns to be $c_2>c_1<p$, it holds that ${\tt score}(p, V)={\tt score}(c_1, p)=2$;
If we add the candidate of $c_3$, vote $v_1$ turns to be $p>c_3>c_1$, the score of $c_1$ is 2 and the score of $p$ is 3.
Therefore, adding $c_3$ can make $p$ being the unique winner of the election.
$\{c_3\}$ is a solution of $3$-CCAC-${\tt Borda}_{\uparrow}$ $(E=(C, C^*, V), \ell)$.

\vbox{}
{\bf Example 5:}Given an instance of $3$-CCDC-${\tt Borda}_{\downarrow}$ $(E=(C, V), \ell)$ where $C=\{c_1, c_2, c_3, p\}$, $V=\{v_1, v_2\}$, $\ell=1$, and $v_1: c_1>c_3>p$, $v_2: p>c_2>c_1$.
Before deleting candidates, the score of $p$'s is equal to $c_1$'s, ${\tt score}(p, V)={\tt score}(c_1, p)=5$.
If we delete the candidate of $c_2$, the score of $p$'s will decrease by 1 corresponding to $c_1$, ${\tt diff}(p, c_1, V)=3-4=-1$;
If we delete the candidate of $c_3$, the score of $c_1$'s will decrease by 1 corresponding to $p$, ${\tt diff}(p, c_1, V)=4-3=1$.
Therefore, deleting $c_3$ can make $p$ being the unique winner of the election.
$\{c_3\}$ is a solution of $3$-CCDC-${\tt Borda}_{\downarrow}$ $(E=(C, V), \ell)$.

\vbox{}
Similar observations can be made for $t$-CCDC-Borda$_{\uparrow}$ with deleting candidates and $t$-CCAC-Borda$_{\downarrow}$ with adding candidates.
Based in these observations and similar reductions from {\sc Dominating Set}, we get the following Theorems.
In the following, we present the details of the hardness proves of $t$-CCAC-${\tt Borda}_{\uparrow}$ and $t$-CCDC-${\tt Borda}_{\downarrow}$ with $t\geq 2$.

\begin{theorem}
\label{thm:top-CCAC-up-w}
For every constant~$t\ge 2$: $t$-CCAC-${\tt Borda}_{\uparrow}$ is NP-hard and W[2]-hard with respect to the number of additions~$\ell$.
\end{theorem}
\begin{proof}
We prove this theorem for $t$-CCAC-Borda$_{\uparrow}$ by a reduction from {\sc Dominating Set}.
Let~$(\mathcal{G}=(\mathcal{V},\mathcal{E}),k')$, where~$|\mathcal{V}|=n'$, be a
{\sc Dominating Set} instance.
We construct a $t$-CCAC-$\delta$ instance~$(E=(C,C^*,V),\ell)$ as follows.
The candidate set~$C$ consists of four subsets: $C=C_1\cup X\cup Y \cup \{p\}$, where~$C_1$ contains~$n'$
candidates~$C_1=\{c_1^1, c_2^1, \ldots, c_{n'}^1\}$, one-to-one corresponding to the
vertices~$v_i'\in \mathcal{V}$. The set~$X_i$ contains~$t_i=n'-{\tt deg}(v_i)-1$ candidates~$x_i^1,\ldots, x_i^{t_i}$ for each vertex~$v_i\in \mathcal{V}$
with~${\tt deg}(v_i')$ denoting the degree of~$v_i'$ in~$\mathcal{G}$, $X=\bigcup_{1\leq i\leq n'}X_i$.
The set~$Y$ contains~$n'$ candidates~$y_1, \ldots, y_{n'}$.
Finally, the set~$C^*$ contains one candidate~$c_i^2$ for each vertex~$v_i'\in \mathcal{V}$.
Then, there are in total~$N=|C_1|+|X|+|Y|+|C'|+1=n'^2+2n'-2m'+1$ candidates.
The vote set~$V$ consists of three subsets~$V=V_1\cup V_2\cup V_3$.
The set~$V_1$ contains a vote of the form: $c_i^2>c_j^1$ for each~$1\le i\le n'$ and each vertex~$v_j'\in (N(v_i')\cup \{v_i'\})$.
The set~$V_2$ contains $t_i=n'-{\tt deg}(v_i)-1$ votes of the form: $c_i^1>x_i^j$ for each~$v_i'\in \mathcal{V'}$, $1\le j\le t_i$.
The set~$V_3$ contains~$n'$ votes of the form: $p>y_j$ with~$1\le j\le n'$.
We set~$\ell=k'$.
The number of ranked candidates in each vote is at most 2, that is~$t=2$.

Before adding the candidates from~$C^*$, each vote in~$V_1$ contains
only one candidate.
Let~${\tt score}^1(c, V)$ and ${\tt score}^2(c, V)$ denote the score~$c$ received from~$V$ before and after adding candidates in~$C'$ to~$C$ respectively. $N^1=|C_1|+|X|+|Y|+1=n'^2+n'-2m'+1$ be the number of candidates before additions.
Then, we have
\begin{equation*}
\begin{split}
  & {\tt score}^1(p, V)={\tt score}^1(c, V)=(N^1-1)n',\ \forall\ c\in C_1,\\
  & {\tt score}^1(c, V)\le N^1-2,\ \forall\ c\in (C_2\cup C_3).
\end{split}
\end{equation*}
Since adding a candidate can only decrease the value of ${\tt diff}(p,c,V)$ by at most 1, candidates in~$C_2\cup C_3$ cannot be the unique winner of the election.
Now, we prove the equivalence between the instances.

\noindent
``$\Longrightarrow$'': Suppose that~$\mathcal{G}$ has a dominating set~$DS$ of size at most~$k'$.
Let~$C^{**}$ be the set of candidates corresponding to the vertices in~$DS$.
Clearly, $|C^{**}|\le \ell$.
Let $N^2=|C_1|+|X|+|Y|+|C^{**}|+1$ be the number of candidates after additions.
Adding~$C^{**}$ to~$C$ changes only the
votes in~$V_1$ and does not change the score of~$p$,
\begin{equation*}
\begin{split}
  {\tt score}^2(p, V)={\tt score}^1(p,V)=(N^2-1)n'.
\end{split}
\end{equation*}
The addition of a candidate~$c_i^2\in C'$ decreases the score of~$c_i^1$ by exactly one due to the vote~$c_i^2>c_j^1$ where a vertex~$v_j'\in N(v_i')$.
The score of~$c_j^1$ received by this vote is changed to~$N^2-2$.
Since for each vertex~$v_i'\in \mathcal{V}$, we have~$\forall v_i'\in DS\ \exists v_j'\in DS$: $v_i'\in N(v_j')$.
Therefore, the scores of all candidates in~$C_1$ are ${\tt score}^2(c_i^1,V)\leq (N^2-2)+(N^2-1)(n'-1)$.
Thus, $p$ is the unique winner of the election.

\noindent
``$\Longleftarrow$'': Suppose that we can add a set~$C^{**}\subseteq C^*$ of at most~$\ell$ candidates into~$C$ to make~$p$
the winner. Let~$\mathcal{V''}$ be the set of vertices corresponding to the candidates in~$C^{**}$. Again, adding candidates from~$C^*$
to~$C$ only change the the structure of votes in $V_1$.
The score of $p$ is ${\tt score}^2(p,V)=(N^2-1)n'$.
Thus, all candidates $c_i^1$ hold ${\tt score}^2(c_i^1,V)\leq (N^2-2)+(N^2-1)(n'-1)$,
that is, for each candidate~$c_i^1$, there is a vote in~$V_1$ of the form: $c_j^2>c_i^1$, with~$c_j^2\in C^{**}$.
By the construction of~$V_1$, the corresponding vertex~$v_i'$ is in~$\mathcal{V''}$ or in
the neighborhood of~$v_j'$. Then,~$\mathcal{V''}$ is a dominating set. This completes the proof of the equivalence and
the W[2]-hardness of $t$-CCAC-Borda$_{\uparrow}$.
\end{proof}

Similar to the prove in theorem~\ref{thm:top-CCAC-up-w}, 
We just do some modifications to the constructions of votes in the following prove of theorem~\ref{thm:top-CCDC-down-w} and make sure:
1) before candidate deletions, some candidates have the same score as $p$'s;
2) after candidate deletions, the decreased score value of candidates whose score is equal to $p$'s before control operation is greater than the decreased score value of $p$'s.
In this way, $p$ will turn to be the unique winner of the election.

\begin{theorem}
\label{thm:top-CCDC-down-w}
For every constant~$t\ge 2$, $t$-CCDC-${\tt Borda}_{\downarrow}$ is NP-hard and W[2]-hard with respect to the number of deletions $\ell$.
\end{theorem}
\begin{proof}
We prove this theorem by a reduction from {\sc dominating set} problem.
Let $(\mathcal{G}=(\mathcal{V},\mathcal{E}),k')$ where $|\mathcal{V}|=n'$ be an instance of dominating set problem.
We construct an instance $(E=(C,V),\ell)$ of $t$-CCDC-${\tt Borda}_{\downarrow}$ according to $(\mathcal{G}=(\mathcal{V},\mathcal{E}),k')$.
For each vertex $v_i'\in \mathcal{V'}$, we construct a candidate $c_i^1$ in $C_1$, $C_1=\{c_1^1,c_2^1,\cdots,c_{n'}^1\}$, a candidate $c_i^2$ in $C_2$, $C_2=\{c_1^2,c_2^2,\cdots,c_{n'}^2\}$, $u_i=n'-{\tt deg}(v_i')-1$ candidates in $X_i$, $X_i=\{x_i^1,x_i^2,\cdots,x_i^{u_i}\}$, $X=\bigcup_{1\leq i\leq n'}X_i$, and one candidate $y_i$ in $Y$, $Y=\{y_1,y_2,\cdots,y_{n'}\}$.
We set $C=C_1\cup C_2\cup X\cup Y\cup \{p\}$.
The vote set $V$ is consist of three subset $V=V_1\cup V_2\cup V_3$.
For each vertex $v_i'\in \mathcal{V'}$, we construct a vote in $V_1$: $c_i^1>c_j^2$ with $v_j'\in N[v_i']$; a vote in $V_2$: $c_i^1>x_i^j$ with $v_j'\notin N[v_{j'}']$; and a vote in $V_3$: $p>y_i$.
Now, the number of ranked candidates in each vote is at most 2, that is $t=2$.
We set $\ell=k'$.

Before deleting candidates, under the rule of ${\tt Borda}_{\downarrow}$, we can get the score of each candidate:\\
(1) Each candidate $c_i^1\in C_1$ receives $2({\tt deg}(v_i')+1)$ points from the votes in $V_1$, ${\tt score}^1(c_i^1,V_1)=2{\tt deg}(v_i')$, receives $2(n'-{\tt deg}(v_i')-1)$ points from $V_2$, ${\tt score}^1(c_i^1,V_2)\\
=2(n'-{\tt deg}(v_i')-1)$, and receives $0$ points from $V_3$, ${\tt score}^1(c_i^1,V_3)=0$.
So, the total score of $c_i^1$ is:
\begin{equation*}
    {\tt score}^1(c_i^1,V) = {\tt score}^1(c_i^1,V_1) + {\tt score}^1(c_i^1,V_2)=2n'.
\end{equation*}
(2) Each candidate $c_i^2\in C_2$ only exists in ${\tt deg}(v_i')+1$ votes in $V_1$, so, the total score of $c_i^2$ is:
\begin{equation*}
    {\tt score}^1(c_i^2,V)={\tt score}^1(c_i^2,V_1)={\tt deg}(v_i')+1.
\end{equation*}
(3) Each candidate in $X\cup Y$ exists once in $V_2\cup V_3$, that is:
\begin{equation*}
    {\tt score}^1(x_i^j,V)={\tt score}^1(y_i,V)=1, 1\leq i\leq n', 1\leq j\leq u_i.
\end{equation*}
(4) The candidate $p$ only exists in the votes of $V_3$, that is: 
\begin{equation*}
    {\tt score}^1(p,V)=2n'.
\end{equation*}
Now, ${\tt score}^1(p,V)={\tt score}^1(c_i^1,V)$, $p$ is not the unique winner.
Next, we prove the equivalence between the two instances.

\noindent
``$\Longrightarrow$'': Suppose that there is size-$\leq k'$ dominating set $DS$ in $\mathcal{G}$.
Let $C^{**}$ be the $C_2$-candidates corresponding to the vertices in $DS$, $C^{**}=\bigcup_{v_i'\in DS}{c_i^2}$.
It is clear that $|C^{**}|=|DS|\leq k'=\ell$.
According to the structure of $V_1$: $c_i^1>c_j^2$ with $v_i'\in N[v_j']$, deleting the candidate $c_j^2$ decreases the scores of some $C_1$ candidates.
Since $DS$ is a dominating set of $\mathcal{G}$, then $\forall v_i^1\in \mathcal{V'}, \exists v_j'\in DS, v_i'\in N(v_j')$.
So the corresponding set $C^{**}$ satisfies $\forall c_i^1\in C_1, \exists c_j^2\in C^{**}, (c_i^1>c_j^2)\in V_1$.
This means the score of candidate $c_i^1\in C_1$ is decreased by at least 1 by deleting the candidates in $C^{**}$, ${\tt score}^2(c_i^1,V)\leq {\tt score}^1(c_i^1,V)-1=2n'-1$.
And, deleting the candidates in $C^{**}$ does not change the score of any other candidates, that is ${\tt score}^2(p,V)={\tt score}^1(p,V)=2n'>{\tt score}^2(c_i^1,V)$.
In this way, $p$ is the unique winner and $C^{**}$ is a solution of this election.

\noindent
``$\Longleftarrow$'': Suppose that there is solution $C^{**}$ of $(E=(C,V),\ell)$ with $|C^{**}|\leq \ell$ that deleting the candidates in $C^{**}$ can make $p$ be the unique winner of the election.
So, ${\tt score}^2(p,V)>{\tt score}^2(c_i^1,V)$.
In addition, since ${\tt score}^1(p,V)={\tt score}^1(c_i^1,V)=2n'$, the score of candidate $c_i^1\in C_1$ is decreased by at least 1 or $c_i^1$ is in $C^{**}$.
According to the structure of $V_1$: $c_i^1>c_j^2$, there must be a vote where the candidate $c_j^2$ or $c_i^1$ is in $C^{**}$.
It means $\forall c_i^1\in C_1\ \exists c_j^2\in C^{**}: (c_i^1>c_j^2)\in V_1$ or $c_i^1\in C^{**}$.
The candidates in $C_1$ and $C_2$ correspond to the vertex in $\mathcal{V}$.
Let $\mathcal{V''}$ be the set of vertices corresponding to the candidates in $C^{**}$.
It is clear that $|\mathcal{V''}|=|C^{**}|\leq \ell=k'$ and $\forall v_i'\in \mathcal{V'}, \exists v_j'\in \mathcal{V''}, v_i'\in N(v_j')$.
Therefore, $\mathcal{V''}$ is a size-$\leq k'$ dominating set of $\mathcal{G}$.
\end{proof}

Next, we continue analyze the complexity of $t$-CCDC-${\tt Borda}_{\uparrow}$ and $t$-CCAC-${\tt Borda}_{\uparrow}$.
For $t$-CCDC-${\tt Borda}_{\uparrow}$ problem, deleting candidate $c$ can increase the scores of candidates behind the position of $c$ corresponding to the other candidates in a vote (the scores of candidates behind $c$ are unchanged while the scores of the other candidates will be decreased).
Similarly, for $t$-CCAC-${\tt Borda}_{\downarrow}$ problem, adding candidate $c$ can increase the scores of candidates before the position of $c$ corresponding to the other candidates (the scores of the candidate before $c$ will increase while the scores of the other candidates will remain unchanged).
The following two examples show the scores of each candidates before and after the control operations.

\vbox{}
{\bf Example 6:}Given an instance of $3$-CCAC-${\tt Borda}_{\downarrow}$ $(E=(C, C^*, V), \ell)$ where $C=\{c_1, p\}$, $C^*=\{c_2, c_3\}$, $V=\{v_1, v_2\}$, $\ell=1$, and $v_1: p>c_3>c_1$, $v_2: c_2>c_1>p$.
Before adding candidates, $v_1$ is $p>c_1$ and $v_2$ is $c_1>p$.
The score of $p$'s is equal to $c_1$'s, ${\tt score}(p, V)={\tt score}(c_1, p)=1$.
If we add the candidate of $c_2$, vote $v_2$ turns to be $c_2>c_1<p$, it holds that ${\tt score}(p, V)={\tt score}(c_1, p)=3$;
If we add the candidate of $c_3$, vote $v_1$ turns to be $p>c_3>c_1$, the score of $c_1$ is 3 and the score of $p$ is 4.
Therefore, adding $c_3$ can make $p$ being the unique winner of the election.
$\{c_3\}$ is a solution of $3$-CCAC-${\tt Borda}_{\downarrow}$ $(E=(C, C^*, V), \ell)$.

\vbox{}
{\bf Example 7:}Given an instance of $3$-CCDC-${\tt Borda}_{\uparrow}$ $(E=(C, V), \ell)$ where $C=\{c_1, c_2, c_3, p\}$, $V=\{v_1, v_2\}$, $\ell=1$, and $v_1: c_1>c_3>p$, $v_2: p>c_2>c_1$.
Before deleting candidates, the score of $p$'s is equal to $c_1$'s, ${\tt score}(p, V)={\tt score}(c_1, p)=4$.
If we delete the candidate of $c_2$, the score of $p$'s will decrease by 1 corresponding to $c_1$, ${\tt diff}(p, c_1, V)=3-4=-1$;
If we delete the candidate of $c_3$, the score of $c_1$'s will decrease by 1 corresponding to $p$, ${\tt diff}(p, c_1, V)=4-3=1$.
Therefore, deleting $c_3$ can make $p$ being the unique winner of the election.
$\{c_3\}$ is a solution of $3$-CCDC-${\tt Borda}_{\downarrow}$ $(E=(C, V), \ell)$.

\vbox{}
In the following, we present the details of the hardness proves of $t$-CCDC-${\tt Borda}_{\uparrow}$ and $t$-CCAC-${\tt Borda}_{\uparrow}$ problems with $t\geq 2$.

\begin{theorem}
\label{thm:top-CCDC-up-w}
For every constant~$t\ge 2$, $t$-CCDC-${\tt Borda}_{\uparrow}$ is NP-hard and W[2]-hard with respect to the number of deletions $\ell$.
\end{theorem}
\begin{proof}
We prove this theorem by a reduction from $D$-regular graph {\sc dominating set} problem.
Let $(\mathcal{G}=(\mathcal{V, \mathcal{E}}), k')$ where $|\mathcal{V}|=n'$ be an instance of $D$-regular graph dominating set problem.
The degree of each vertex is $D$ in $D$-regular graph.
Then we construct an instance $(E=(C,V),\ell)$ of $t$-CCDC-${\tt Borda}_{\uparrow}$ problem according to $(\mathcal{G}=(\mathcal{V', \mathcal{E'}}), k')$.
For each vertex $v_i'\in \mathcal{V'}$, we construct a candidate $c_i^1$ in $C_1$, that is $C_1=\{c_1^1, c_2^1, \cdots, c_{n'}^1\}$, a candidate $c_i^2$ in $C_2$, that is $C_2=\{c_1^2, c_2^2, \cdots, c_{n'}^2\}$, and $D+1$ candidates in $Y$, that is $Y_i=\{y_i^1,y_i^2,\cdots, y_i^{D+1}\}$, $Y=\bigcup_{1\leq i\leq n'}Y_i$.
We set $C=C_1\cup C_2 \cup Y\cup \{p\}$, $N^1=|C_1|+|C_2|+|Y|+|\{p\}|=(D+3)*n+1$.
In the following, we construct the set of votes $V$.
For each vertex $v_i'$, we construct a vote in $V_1$ of the form: $c_i^1>c_j^2$ with  $v_j'\in N[v_i']$; a vote in $V_2$ of the form: $c_{j'}^2>c_i^1$ with  $v_{j'}'\notin N[v_i']$; $D+3$ votes in $V_3$ of the form: $y_i^j>c_i^1$ with $1\leq j\leq D+1$, and a vote in $V_4$ of the form: $c_i^2>p$.
In addition, we construct $n'$ identical votes in $V_5$ of the form: $p$.
We set $V=V_1\cup V_2 \cup V_3 \cup V_4 \cup V_5$ and $\ell=k'$.
The number of ranked candidates in each vote is at most 2, that is $t=2$.

Before deleting candidates, each candidate $c_i^1 \in C_1$ receives $(N^1+1)(D+1)$ points from $V_1$ and $(N^1-1)n'$ points from $V_2\cup V_3$, that is ${\tt score}^1(c_i^1, V)=(N^1-1)\times 2n'$.
Each candidate $c_i^2\in C_2$ receives $(N^1-2)(D+1)$ points from $V_1$, $(N^1-1)(n'-D-1)$ points from $V_2$, and $(N^1-1)n'$ points from $V_4$, that is ${\tt score}^1(c_i^2,V)=(N^1-2)(D+1)+(N^1-1)(n'-D-1)+(N^1-1)n'$.
Each candidate in $Y$ only exists once in $V_3$, ${\tt score}^1(y_i,V)=N^1-1$.
Candidate $p$ receives $(N^1-1)n'$ points from $V_4$ and $(N^1-1)n'$ points from $V_5$, that is ${\tt score}^1(p,V)={\tt score}^1(p,V_4)+{\tt score}^1(p,V_5)=(N^1-1)\times 2n'$.
Now, ${\tt score}^1(p,V)={\tt score}^1(c_i^1,V)$, $p$ is not the unique winner.

According to the construction of $V$, the value of ${\tt diff}(p,c_i^1,V)$ is increased only when some $C_2$-candidates are deleted.
Deleting the candidates in $Y$ only increases the scores of some $C_1$ candidates.
Therefore, if $C^{**}$ is a solution of $(E=(C,V),\ell)$ containing a candidate $y_i^j\in Y$, $C^{**}\setminus{\{y_i^j\}}$ is also a solution of $(E=(C,V),\ell)$.
In this way, each solution of $(E=(C,V),\ell)$ can be transformed to another solution with only candidates in $C_1\cup C_2$.
So, we just analyze the conditions of eliminating candidates in $C_1\cup C_2$ here.
In the following, we prove the equivalence between $(\mathcal{G}=(\mathcal{V, \mathcal{E}}), k')$ and $t$-CCDC-${\tt Borda}_{\uparrow}$.

\noindent
``$\Longrightarrow$'': Suppose that there is a size-$\leq k'$ dominating set $DS$.
Let $C^{**}$ be the set of $C_2$-candidates corresponding to the vertices in $DS$, $C^{**}=\bigcup_{v_i'\in DS}\{c_i^2\}$.
It is obviously that $|C^{**}|=|DS|\leq k'=\ell$.
According to the construction of $V_4$: $c_i^2>p$, deleting the candidates in $C^{**}$ can increase the score of $p$ by $|C^{**}|$ corresponding to the other candidates. 
And, according to the structure of $V_2$: $c_j^2>c_i^1,  v_i'\notin N[v_j']$, deleting a candidate $c_j^2$ of $C_2$ will increase the score of some $C_1$ candidates corresponding to the other candidates.
Similarly, according to the structure of $V_1$: $c_i^1>c_j^2, v_i'\in N[v_j']$, deleting any candidate $c_j^2$ of $C_2$ cannot increase the score of any $C_1$ candidate.
Since $DS$ is a dominating set of $\mathcal{G}$, that is $\forall v_i'\in \mathcal{V'}, \exists v_j'\in DS, v_i'\in N[v_j']$.
So, the corresponding candidates of $C^{**}$ hold: $\forall c_i^1\in C_1, \exists c_j^2\in C^{**}: (c_i^1>c_j^2)\in V_1$.
It means the score of each $C_1$ candidate is at most $2n'(N^2-1)+|C^{**}|-1$ points after deleting the candidates in $C^{**}$, that is ${\tt score}^2(c_i^1,p)\leq 2n'(N^2-1)+|C^{**}|-1$.
Since, it always be ${\tt score}^2(p,V)=2n'(N^2-1)+|C^{**}|$.
It holds ${\tt diff}(p,c_i^1,V)\geq 1$.
In this way, $p$ is the unique winner of the election.

\noindent
``$\Longleftarrow$'': Suppose there is a solution $C^{**}$ of the election $(E=(C,V),\ell)$ with $|C^{**}|\leq \ell$ that deleting the candidates in $C^{**}$ can make $p$ be the unique winner.
So, there must be ${\tt score}^2(p,V)>{\tt score}^2(c_i^1,V)$.
It means the decreased score value of $p$ is smaller than the one of each candidate $c_i^1\in C_1$.
Since each solution $C^{**}$ of $(E=(C,V),\ell)$ can transformed to another solution only containing the candidates in $C_1\cup C_2$.
So, $C^{**}$ only contains the candidates in $C_1\cup C_2$.
According to the structure of $V_1$: $c_i^1>c_j^2$, the candidate $c_j^2$ or $c_i^1$ must be in $C^{**}$, that is $\forall c_i^1\in C_1, \exists c_j^2\in C^{**}, (c_i^1>c_j^2)\in V_1$ or $c_i^1\in C^{**}$.
The candidates of $C_1$ and $C_2$ are corresponding to the vertices in $\mathcal{V}$.
Let $\mathcal{V''}$ be the vertices corresponding to $C^{**}$.
It is obvious that $|\mathcal{V''}|=|C^{**}|\leq \ell=k'$ and $\forall v_i'\in \mathcal{V}, \exists v_j'\in \mathcal{V''}, v_i'\in N[v_j']$.
Therefore, $\mathcal{V''}$ is a size-$\leq k'$ dominating set of $\mathcal{G}$.
\end{proof}

Similar to the prove in theorem~\ref{thm:top-CCDC-up-w}, 
We just do some modifications to the constructions of votes in the following prove of theorem~\ref{thm:top-CCAC-down-w} and make sure:
1) before candidate deletions, some candidates have the same score as $p$'s;
2) after candidate deletions, the decreased score value of candidates whose score is equal to $p$'s before control operation is greater than the decreased score value of $p$'s.
In this way, $p$ will turn to be the unique winner of the election.

\begin{theorem}
\label{thm:top-CCAC-down-w}
For every constant~$t\ge 2$, $t$-CCAC-${\tt Borda}_{\downarrow}$ is NP-hard and W[2]-hard with respect to the number of additions $\ell$.
\end{theorem}
\begin{proof}
We prove this theorem by a reduction from {\sc dominating set} problem.
Let $(\mathcal{G}=(\mathcal{V},\mathcal{E}),k')$ be an instance of {\sc dominating set} problem where $|\mathcal{V}|=n'$.
We construct a $t$-CCAC-${\tt Borda}_{\downarrow}$ instance $(E=(C,C^*,V),\ell)$ according to $(\mathcal{G}=(\mathcal{V},\mathcal{E}),k')$.
For each vertex $v_i'\in \mathcal{V}$, we construct a candidate $c_i^1$ in $C_1$, $C_1=\{c_1^1,c_2^1,\cdots,c_{n'}^1\}$, and a candidate $c_i^2$ in $C_2$, $C_2=\{c_1^2,c_2^2,\cdots,c_{n'}^2\}$.
Let $C=C_1\cup \{p\}$ and $C^*=C_2$.
The vote set $V$ is consist of three component $V=V_1\cup V_2 \cup V_3$.
For each vertex $v_i'\in \mathcal{V}$, we construct a vote in $V_1$ of the form: $c_i^1>c_j^2$ with  $v_j'\notin N[v_i']$; ${\tt deg}(v_i')+1$ votes in $V_2$ of the form: $c_i^1$, and a vote in $V_3$ of the form: $p>c_i^2$.
The number of ranked candidates in each vote is at most 2, that is, $t=2$.
We set $\ell:=k'$.

Before adding candidates in $C^*$, each vote $v\in V_1$ ranks only one candidate: $c_i^1$, and each candidate $c_i^1$ is in exact $n'-{\tt deg}(v_i')-1$ votes, ${\tt score}^1(c_i^1,V_1)=n'-{\tt deg}(v_i')-1$.
Similarly, each vote $v\in V_2$ also rankss only one candidate: $c_i^1$, and each candidate $c_i^1$ is in ${\tt deg}(v_i')+1$ votes, ${\tt score}^1(c_i^1,V_2)={\tt deg}(v_i')+1$.
Each $V_3$ vote ranks a candidate $p$, ${\tt score}^1(p,V_3)=n'$.
So, before adding candidates in $C^*$, the score of each candidates are as follows:
\begin{equation*}
\begin{split}
{\tt score}^1(c_i^1,V)&={\tt score}^1(c_i^1,V_1)+{\tt score}^1(c_i^1,V_2)\\
&=n'-{\tt deg}(v_i')-1+{\tt deg}(v_i')+1=n';\\
{\tt score}^1(p, V)&={\tt score}^1(p,V_3)=n'.
\end{split}
\end{equation*}
So, candidate $p$ is not the unique winner of the election.
Under the rule of $Borda_{\downarrow}$, adding candidates can only increase the score of some candidates.
In addition, the candidates of $C^*$ only exist in $V_1\cup V_3$.
Therefore, we only analyze the votes in $V_1\cup V_3$ after adding candidates.
In the following, we prove the equivalence between the two problems.

\noindent
``$\Longrightarrow$'': Suppose that there is a dominating set $DS$ of $\mathcal{G}$ with $|DS|\leq k'$.
Let $C^{**}$ be the set of candidates in $C_2$ corresponding to the vertices in $DS$, that is $C^{**}=\bigcup_{v_i'\in DS}\{c_i^2\}$.
It is clear that $|C^{**}|=|DS|\leq k'=\ell$.
According to the structure of $V_3$: $p>c_i^2$, adding $|C^{**}|$ candidates can increase the score of $p$ by $|C^{**}|$, ${\tt score}^2(p,V)={\tt score}^1(p,V)+|C^{**}|=n'+|C^{**}|$.
Similarly, according to the structure of $V_1$: $c_i^1>c_j^2$ with $v_j'\notin N[v_i']$, adding the candidate $c_j^2$ can increase the score of $c_i^1$ by 1.
Since $DS$ is a dominating set of $\mathcal{G}$, then $\forall v_i'\in \mathcal{V}, \exists v_j'\in DS, v_i'\in N[v_j']$.
The candidates in $C^*$ correspond to the vertices in $\mathcal{V}$.
It satisfies $\forall c_i^1\in C_1, \exists c_j^2\in C^{**}, (c_i^1>c_j^2)\notin V_1$.
This means the score of each candidate $c_i^1$ is increased by at most $|C^{**}|-1$, that is ${\tt score}^2(c_i^1,V)\leq {\tt score}^1(c_i^1,V)+|C^{**}|-1=n'+|C^{**}|-1$.
The score of each $C_2$ candidate is ${\tt score}^2(c_j^2,V)={\tt score}^2(c_j^2,V_1)+{\tt score}^2(c_j^2,V_3)=n'-{\tt deg}(v_i')-1+1<n'$.
Therefore, after adding the candidates in $C^{**}$, it holds ${\tt diff}(p,c_j^2,V)>{\tt diff}(p,c_i^1,V)>0$.
$C^{**}$ is a solution of $(E=(C,C^{*},V),\ell)$.

\noindent
``$\Longleftarrow$'': Suppose that there is no size-$\leq k'$ dominating set in $\mathcal{G}$.
So for each vertex set $\mathcal{V''}$ with $|\mathcal{V''}|\leq k'$, there must be a candidate $v'$ not dominated by the vertex in $\mathcal{V''}$.
That is $\exists v_i'\in \mathcal{V}, \forall v_j'\in \mathcal{V''}, v_i'\notin N[v_j']$.
Let $C^{**}$ be the set of candidates corresponding to the vertices in $\mathcal{V''}$.
According to the structure of $V_1$: $c_i^1>c_j^2$ with $v_j'\notin N[v_i']$, there must be a candidate $c_i^1$ which is in a vote of $V_1$ together with a candidate $c_j^2$ in $C^{**}$.
It means $\exists c_i^1\in C_1, \forall c_j^2\in C^{**}, (c_i^1>c_j^2)\in V_1$.
So, there must be a candidate whose score is increased by $|C^{**}|$.
Since the score of $p$ is also increase by $|C^{**}|$ when adding all candidates in $C^{**}$.
In this way, there must be a candidate $c_i^1$ satisfying ${\tt score}^2(p,V)={\tt score}^2(c_i^1,V)=n'+|C^{**}|$.
Therefore, the election $(E=(C,C^*,V),\ell)$ instance does not have a subset $C^{**}$ with $|C^{**}|\leq \ell$, and eliminating the candidates in $C^{**}$ can make $p$ being the unique winner of the election.
\end{proof}

For the control problems under the rule of ${\tt Borda}_{av}$, the score value changes of candidates after the control operations are the same as under the rule of ${\tt Borda}_{\uparrow}$ (adding/del-
eting candidates can only increase/decrease the scores of some candidates and the scores of the other candidates remain unchanged).
Let ${\tt diff}_{av}(c,c',V)$ be the value of ${\tt diff}(c,c',V)$ under the rule of ${\tt Borda}_{av}$, ${\tt diff}_{\uparrow}(c,c',V)$ be the value of ${\tt diff}(c,c',V)$ under the rule of ${\tt Borda}_{\uparrow}$.
For any both ranked (or both unranked) candidates $c$ and $c'$ of a vote $v$, the value of ${\tt diff}(c,c',\{v\})$ is the same under the rule of ${\tt Borda}_{\uparrow}$ and ${\tt Borda}_{av}$, that is, ${\tt diff}_{\uparrow}(c,c',\{v\})={\tt diff}_{av}(c,c',\{v\})$; and for a ranked candidate $c$ and an unranked candidate $c'$ of a vote $v$, the value of ${\tt diff}_{av}(c,c',\{v\})-{\tt diff}_{\uparrow}(c,c',\{v\})$ is a constant, that is, ${\tt diff}_{av}(c,c',\{v\})-{\tt diff}_{\uparrow}(c,c',\{v\})=\frac{m-|v|-1}{2}$.
Thus, by adding some dummy candidates, we can make sure the values of ${\tt diff}_{\uparrow}(c,c',\\
\{v\})$ and ${\tt diff}_{av}(c,c',\{v\})$ are always the same under the rules of ${\tt Borda}_{\uparrow}$ and ${\tt Borda}_{av}$.
We can prove the following theorem.

\begin{theorem}
\label{top-CCAC-CCDC-av-w}
For every constant~$t\ge 2$, $t$-CCAC-${\tt Borda}_{av}$ and $t$-CCDC-${\tt Borda}_{av}$ are NP-hard and W[2]-hard with respect to the number of operations $\ell$.
\end{theorem}
\begin{proof}
We prove this theorem by reductions from $t$-CCAC-${\tt Borda}_{\uparrow}$ and $t$-CCDC-${\tt Borda}_{\uparrow}$ problems.
Each instance of $t$-CCAC-${\tt Borda}_{\uparrow}$ and $t$-CCDC-${\tt Borda}_{\uparrow}$ can be transformed into an instance of $t$-CCAC-${\tt Borda}_{av}$ and $t$-CCDC-${\tt Borda}_{av}$ in polynomial time respectively.
Given an instance $(E^1=(C^1,C^{1*},V^1),\ell^1)$ of $t$-CCAC-${\tt Borda}_{\uparrow}$ problem, we construct an instance $(E^2=(C^2,C^{2*},V^2),\ell^2)$ of $t$-CCAC-${\tt Borda}_{av}$.
We construct a set of dummy candidates $C'$ with $|C'|=|C^1|-t-1$.
Let $C^2=C^1\cup C'$, $C^{2*}=C^{1*}$, $V^2=V^1$, and $\ell^2=\ell^1$.
$E^1$ and $E^2$ are the same except for the candidates in $C'$.
The dummy candidates in $C'$ do not exist in any votes.
The aim of construction these dummy candidates is to control the scores of $C^1$ in $E^2$.
In the following, we prove that any two candidates $c$ and $c'$ in $C^1$ satisfying ${\tt diff}(c,c',V^1)={\tt diff}(c,c',V^2)$.

Let $N^1$ be the number of candidates in $C^1$, and $N^2$ be the number of candidates in $N^2$.
It holds $N^1=|C^1|$, $N^2=|C^1|+|C'|=2|C^1|-t-1$.
Under the rule ${\tt Borda}_{\uparrow}$, if ${\tt pos}(c,v)\neq 0$, the candidate $c$ receives $N^1-{\tt pos}(c,v)$ points from $v$, otherwise, $c$ receives $0$ points from $v$.
Let $V^1(c)$ be the set of votes where $c$ is a ranked candidate in a vote of $V^1$, that is, $V^1(c)=\{v\ |\ {\tt pos}(c,v)\neq 0, v\in V^1\}$.
The total score of $c$ is: ${\tt score}(c,V^1)=\sum_{v\in V^1(c)}(N^1-{\tt pos}(c,v))$.
Under the rule ${\tt Borda}_{av}$, if ${\tt pos}(c,v)\neq 0$, the candidate $c$ receives $N^2-{\tt pos}(c,v)$ points from $v$, otherwise, $c$ receives $\frac{N^2-t-1}{2}$ points from $v$.
Similarly, Let $V^2(c)$ be the set of votes where $c$ is a ranked candidate in a vote of $V^2$, that is, $V^2(c)=\{v\ |\ {\tt pos}(c,v)\neq 0, v\in V^2\}$.
The total score of $c$ is ${\tt score}(c,V^2)=\sum_{v\in V^2(c)}(N^2-{\tt pos}(c,v))+|V^2\setminus{V^2(c)}|\times \frac{N^2-t-1}{2}$.
Then, for any two candidates $c$ and $c'\in C^1$ of $E^1$, it holds:
\begin{equation*}
\begin{split}
   {\tt diff}(c,c',V^1)&=\sum_{v\in V^1(c)}(N^1-{\tt pos}(c,v))\\
   &-\sum_{v\in V^1(c')}(N^1-{\tt pos}(c',v))\\
   &=(|V^1(c)|-|V^1(c')|)\times N^1\\
   &-(\sum_{v\in V^1(c)}({\tt pos}(c,v))-\sum_{v\in V^1(c')}({\tt pos}(c',v))).
\end{split}
\end{equation*}
For $E^2$, the dummy candidate cannot be the unique winner.
We just calculate the score of candidates in $C^1$.
For any two candidates $c$ $c'\in C^1$ of $E^2$, it holds:
\begin{equation*}
\begin{split}
{\tt score}(c,V^2)&=\sum_{v\in V^2(c)}(N^2-{\tt pos}(c,v))\\
&+|V^2\setminus{V^2(c)}|\times \frac{N^2-t-1}{2},\\
{\tt score}(c',V^2)&=\sum_{v\in V^2(c')}(N^2-{\tt pos}(c',v))\\
&+|V^2\setminus{V^2(c')}|\times \frac{N^2-t-1}{2},\\
   {\tt diff}(c,c',V^2)&={\tt score}(c,V^2)-{\tt score}(c',V^2)\\
   &=(|V^1(c)|-|V^1(c')|)\times(N^2-\frac{N^2-t-1}{2})\\
   &-(\sum_{v\in V^1(c)}({\tt pos}(c,v))-\sum_{v\in V^1(c')}({\tt pos}(c',v)))\\
    &=(|V^1(c)|-|V^1(c')|)\times N^1\\
    &-(\sum_{v\in V^1(c)}({\tt pos}(c,v))-\sum_{v\in V^1(c')}({\tt pos}(c',v)))\\
    &={\tt diff}(c,c',V^1).
\end{split}
\end{equation*}
It means that the different scores of any two candidates $c,c'\in C^1$ in $E^1$ and $E^2$ are the same.
Since $C^{2*}=C^{1*}$, $V^2=V^1$, $\ell^2=\ell^1$, then $E^2$ has a solution if and only if $E^1$ has a solution.
According theorem~\ref{thm:top-CCAC-up-w}, $t$-CCAC-${\tt Borda}_{\uparrow}$ problem is NP-hard and W[2]-hard with respect to~$a$,
So, $t$-CCAC-${\tt Borda}_{av}$ problem is NP-hard and W[2]-hard with respect to~$\ell$.

Similarly, each instance $(E^3=(C^3,V^3),\ell^3)$ of $t$-CCDC-${\tt Borda}_{\uparrow}$ can be transformed to an instance $(E^4=(C^4,V^4),\ell^4)$ of $t$-CCDC-${\tt Borda}_{av}$ with $C^4=C^3\cup C'$, $V^4=V^3$, and $\ell^4=\ell^3$ where $C'$ contains $|C^3|-t-1$ dummy candidates.
In $E^3$ and $E^4$,  it holds ${\tt diff}(c,c',V^3)={\tt diff}(c,c',V^4)$ for any two candidates $c$ and $c'\in C^3$,.
$E^4$ has a solution if and only if $E^3$ has a solution.
According theorem~\ref{thm:top-CCDC-up-w}, $t$-CCDC-${\tt Borda}_{\uparrow}$ problem is NP-hard and W[2]-hard with respect to~$\ell$,
So, $t$-CCDC-${\tt Borda}_{av}$ problem is NP-hard and W[2]-hard with respect to~$\ell$.
The proof is done.
\end{proof}

\section{Conclusion}
In this paper, we complete the parameterized complexity analysis of all Borda control by adding/deleting votes/candidates problems.
In particular, we show W[2]-hardness of CCDV with respect to the number of deletions, solving an open question posed in~\cite{LZ-TCS-2013}.
Furthermore, we also achieve a complete picture of classical and parameterized complexity of these control problems with $t$-truncated votes.
It is interesting to observe that the classical complexity border lies between $t=2$ and $t=3$ for CCAV and CCDV, and between $t=1$ and $t=2$ for CCAC and CCDC.
Meanwhile, CCAV and CCDV are in FPT with respect to $t$ and the number $\ell$ of additions/deletions, while CCAC and CCDC are W[2]-hard with the number of additions/deletions as parameter even with $t=2$.
For the future work, it is interesting to examine whether a dichotomy result concerning parameterized complexity for scoring rule control problems in the sense of~\cite{ELH-NCAI-2014} can be achieved.
Hereby, the number of additions/deletions seems a reasonable parameter.
The basic idea behind the reduction in the proof of Theorem~\ref{Complete-CCDV-ell} might also work for other scoring rules.
Furthermore, Fitzsimmons and Hemaspaandra~\cite{FH-ICAD-2015} and Menon and Larson~\cite{ML-AAMAS-2017} consider a large collection of manipulation, control, bribery problems in relation to truncated votes.
Hereby, very few problems have been examined from the viewpoint of parameterized complexity.
Both algorithmic ideas and reduction gadgets used here could be applicable to these truncated votes problems.
Finally, we are not aware of any parameterized complexity result for Borda control by partitioning votes or candidates or applying other operations in~\cite{NR-AAAI-2017}.
\clearpage


\begin{thebibliography}{10}
\providecommand{\url}[1]{\texttt{#1}}
\providecommand{\urlprefix}{URL }
\providecommand{\doi}[1]{https://doi.org/#1}

\bibitem{A-A-2010}
Ailon, N.: Aggregation of partial rankings, $p$-ratings and top-$m$ lists. Algorithmica pp. 284--300 (2010)

\bibitem{BTT-MCM-1992}
Bartholdi, J., Tovey, C., Trick, M.: How hard is it to control an election? Mathematical and Computer Modelling pp. 27--40 (1992)

\bibitem{BFLR-AAMAS-2012}
Baumeister, D., Faliszewski, P., Lang, J., Rothe, J.: Campaigns for lazy voters: truncated ballots. In: AAMAS. pp. 577--584 (2012)

\bibitem{DJ12}
Baumeister, D., Rothe, J.: Taking the final step to a full dichotomy of the possible winner problem in pure scoring rules. Inf. Process. Lett. pp. 186--190 (2012)

\bibitem{BR-EC-2016}
Baumeister, D., Rothe, J.: Preference aggregation by voting. In: EC, pp. 197--325 (2016)

\bibitem{BS-MPCO-2009}
Brams, S., Sanver, R.: Voting systems that combine approval and preference. In: The Mathematics of Preference, Choice and Order, pp. 215--237 (2009)

\bibitem{BER-AAMAS-2016}
Briskorn, D., Erd{\'{e}}lyi, G., Reger, C.: Bribery in k-approval and k-veto under partial information: (extended abstract). In: AAMAS. pp. 1299--1300 (2016)

\bibitem{CDKKRS-20}
Chakraborty, V., Delemazure, T., Kimelfeld, B., Kolaitis, P.G., Relia, K., Stoyanovich, J.: Algorithmic techniques for necessary and possible winners. CoRR  (2020)

\bibitem{VP21}
Chakraborty, V., Kolaitis, P.G.: Classifying the complexity of the possible winner problem on partial chains. In: AAMAS. pp. 297--305 (2021)

\bibitem{CFNT-AAAI-2015}
Chen, J., Faliszewski, P., Niedermeier, R., Talmon, N.: Elections with few voters: Candidate control can be easy. J. Artif. Intell. Res. pp. 937--1002 (2017)

\bibitem{CFKLMPPS-Springer-2015}
Cygan, M., Fomin, F.V., Kowalik, {\L}., Lokshtanov, D., Marx, D., Pilipczuk, M., Pilipczuk, M., Saurabh, S.: Parameterized Algorithms. Springer (2015)

\bibitem{DF99}
Downey, R.G., Fellows, M.R.: Parameterized Complexity. Monographs in Computer Science (1999)

\bibitem{EFS-AIR-2011}
Elkind, E., Faliszewski, P., Slinko, A.: Cloning in elections: Finding the possible winners. J. Artif. Intell. Res. pp. 529--573 (2011)

\bibitem{EFRS-JCS-15}
Erd{\'{e}}lyi, G., Fellows, M.R., Rothe, J., Schend, L.: Control complexity in bucklin and fallback voting: {A} theoretical analysis. J. Comput. Syst. Sci. pp. 632--660 (2015)

\bibitem{FR2016}
Faliszewski, P., Rothe, J.: Control and bribery in voting. In: Handbook of Computational Social Choice, pp. 146--168 (2016)

\bibitem{FH-ICAD-2015}
Fitzsimmons, Z., Hemaspaandra, E.: Complexity of manipulative actions when voting with ties. In: ADT. pp. 103--119 (2015)

\bibitem{ELH-NCAI-2014}
Hemaspaandra, E., Hemaspaandra, L., Schnoor, H.: A control dichotomy for pure scoring rules. AAAI pp. 712--720 (2014)

\bibitem{HHR-AI-2007}
Hemaspaandra, E., Hemaspaandra, L.A., Rothe, J.: Anyone but him: The complexity of precluding an alternative. J. Artif. Intell. Res. pp. 255--285 (2007)

\bibitem{HS-ECAI-2016}
Hemaspaandra, E., Schnoor, H.: Dichotomy for pure scoring rules under manipulative electoral actions. In: ECAI. pp. 1071--1079 (2016)

\bibitem{B-AAMAS-19}
Kenig, B.: The complexity of the possible winner problem with partitioned preferences. In: AAMAS. pp. 2051--2053 (2019)

\bibitem{LZ-TCS-2013}
Liu, H., Zhu, D.: Parameterized complexity of control by voter selection in {Maximin}, {Copeland}, {Borda}, {Bucklin}, and {Approval} election systems. Theoretical Computer Science pp. 115--123 (2013)

\bibitem{LNRVW-AAMAS-2015}
Loreggia, A., Narodytska, N., Rossi, F., Venable, K., Walsh, T.: Controlling elections by replacing candidates or votes. In: AAMAS. pp. 1737--1738 (2015)

\bibitem{ML-AAMAS-2017}
Menon, V., Larson, K.: Computational aspects of strategic behaviour in elections with top-truncated ballots. In: AAMAS. pp. 1506--1547 (2017)

\bibitem{NW-ECAI-2014}
Narodytska, N., Walsh, T.: The computational impact of partial votes on strategic voting. In: ECAI. pp. 657--662 (2014)

\bibitem{NR-AAAI-2017}
Neveling, M., Rothe, J.: Solving seven open problems of offline and online control in {Borda} elections. In: AAAI. pp. 3029--3035 (2017)

\bibitem{R-RIT-2007}
Russel, N.: Complexity of control of {Borda} count elections. Master's thesis, Rochester Institute of Technology  (2007)

\bibitem{LV11}
Xia, L., Conitzer, V.: Determining possible and necessary winners given partial orders. J. Artif. Intell. Res. pp. 25--67 (2011)

\bibitem{Y-ECAI-2014}
Yang, Y.: Election attacks with few candidates. In: ECAI. pp. 1131--1132 (2014)

\bibitem{Y-AAMAS-2017}
Yang, Y.: On the complexity of {Borda} control in single-peaked elections. In: AAMAS. pp. 1178--1186 (2017)

\end{thebibliography}

\end{document}